%% file: psfe.tex
\documentclass[12pt, letterpaper, leqno]{article}
\usepackage[margin=1in,footskip=0.25in]{geometry}

\usepackage{amsmath, amssymb}
\usepackage{amsthm}
\usepackage{mathtools}
\allowdisplaybreaks

\usepackage[labelfont=bf]{caption}

\usepackage[T1]{fontenc}
\usepackage[scale=1.07]{cochineal}
\usepackage[vvarbb]{newtxmath}
\usepackage[scaled=0.8]{beramono}
\usepackage{sectsty}
\sectionfont{\large}

\usepackage{bm}

\DeclareSymbolFont{largesymbolsCM}{OMX}{cmex}{m}{n}

\let\sum\relax
\DeclareMathSymbol{\sum}{\mathop}{largesymbolsCM}{"50}

\DeclareSymbolFont{lettersCM}{OML}{cmm} {m}{it}
\SetSymbolFont{lettersCM}{bold}{OML}{cmm} {b}{it}

\DeclareMathSymbol{\pi}{\mathord}{lettersCM}{"19}

\DeclareSymbolFont{symbolsCM}{OMS}{cmsy}{m}{n}
\SetSymbolFont{symbolsCM}{bold}{OMS}{cmsy}{b}{n}

\DeclareMathSymbol{\infty}{\mathord}{symbolsCM}{"31}
\DeclarePairedDelimiter\abs{\lvert}{\rvert}%

\usepackage{xcolor}
\definecolor{linkcolor}{HTML}{0645AD}
\usepackage{hyperref}
\hypersetup{
     colorlinks   = true,
     allcolors    = linkcolor,
}

\definecolor{white}{RGB}{255,255,255}
\definecolor{crimson}{RGB}{220,20,60}
\definecolor{blue}{RGB}{0,0,205}
\definecolor{myblue}{RGB}{80,80,160}
\definecolor{mygreen}{RGB}{80,160,80}
\colorlet{tcrimson}{white!40!crimson}
\colorlet{tblue}{white!40!blue}
\definecolor{porange}{HTML}{EE7F2D}

\definecolor{simblue}{HTML}{00798c}
\definecolor{simgreen}{HTML}{66a182}
\definecolor{simgrey}{HTML}{b3b3b3}

\def\*#1{\mathbf{#1}}

\DeclareMathOperator{\interior}{int}

\newcommand{\argmax}{\mathop{\rm arg~max}\limits}
\newcommand{\argmin}{\mathop{\rm arg~min}\limits}

\newcommand{\norm}[1]{\lVert#1\rVert}

\newcommand{\inprob}{\overset{p}{\to}}

\newcommand{\E}{\mathbb{E}}

\newcommand{\V}{\mathbb{V}}

\newcommand{\hP}{\mathbb{P}}

\newcommand{\indep}{\!\perp\!\!\!\perp}

\newtheorem{theorem}{Theorem}
\newtheorem{proposition}{Proposition}

\newtheorem{lemma}{Lemma}
\newtheorem{assumption}{Assumption}

\usepackage{enumitem}

\numberwithin{equation}{section}

\usepackage{setspace}
\usepackage{pdfpages}
\usepackage{booktabs}

\usepackage{float}

\usepackage{tikz}

\usepackage{natbib}

\newcommand{\IPTWFE}{\texttt{IPTW-FE}} 
\newcommand{\IPTWT}{\texttt{IPTW-True}} 
\newcommand{\IPTW}{\texttt{IPTW}} 

\begin{document}

\title{\textbf{Adjusting for Unmeasured Confounding in Marginal Structural Models with Propensity-Score Fixed Effects}\thanks{Thanks to Adam Glynn for extensive discussions and feedback. We also thank Gary King for generous comments. Any errors remain our own. Comments welcome.}}
\author{
Matthew Blackwell\thanks{
  Associate Professor, Department of Government, Harvard University.
  Email:~\href{mblackwell@gov.harvard.edu}{mblackwell@gov.harvard.edu}.
  Web:~\url{www.mattblackwell.org} }%
\and%
Soichiro Yamauchi\thanks{Graduate student, Department of Government, Harvard University.
  Email:~\href{mailto:syamauchi@g.harvard.edu}{syamauchi@g.harvard.edu}.
  Web:~\url{https://soichiroy.github.io}}
  }
\date{\today}
\maketitle

\setstretch{1.2}
\begin{abstract}
Marginal structural models are a popular tool for investigating the effects of time-varying treatments, but they require an assumption of no unobserved confounders between the treatment and outcome. With observational data, this assumption may be difficult to maintain, and in studies with panel data, many researchers use fixed effects models to purge the data of time-constant unmeasured confounding. Unfortunately, traditional linear fixed effects models are not suitable for estimating the effects of time-varying treatments, since they can only estimate lagged effects under implausible assumptions. To resolve this tension, we a propose a novel inverse probability of treatment weighting estimator with propensity-score fixed effects to adjust for time-constant unmeasured confounding in marginal structural models of fixed-length treatment histories. We show that these estimators are consistent and asymptotically normal when the number of units and time periods grow at a similar rate. Unlike traditional fixed effect models, this approach works even when the outcome is only measured at a single point in time as is common in  marginal structural models. We apply these methods to estimating the effect of negative advertising on the electoral success of candidates for statewide offices in the United States.
\end{abstract}

\clearpage

\setstretch{1.75}

\section{Introduction}\label{s:introduction}

Researchers---especially social scientists---often find themselves estimating causal effects in complex, dynamic settings. For example, when studying the effects of political campaign strategies,  researchers must confront how candidates change their messaging and behavior in response to the shifting landscape of the campaign. While most social scientists ignore these dynamics, a handful of scholars have applied the marginal structural model (MSM) approach to estimating the effects of time-varying treatments \citep{RobHerBru00} to a variety of social science settings  \citep{Blackwell13, LadHarWin18, CreSim19, Kurer20}. Given the richness of dynamic phenomena in the social sciences, it is perhaps surprising that such methods are not applied more widely in those fields.  One key reason is that most approaches to estimating the effects of these treatments---including marginal structural models, but also parametric g-computation and structural nested mean models---require a no unmeasured confounding assumption that many social scientists find implausible. In this paper, we extend the current MSM framework to allow for time-constant unmeasured confounding and apply this approach to estimating the time-varying effects of negative advertising in U.S.\ elections.

Marginal structural models have been a popular way to address the key challenges that time-varying treatments pose, especially in the biomedical sciences. The unique challenges of these treatments include the presence of time-varying confounders, or those variables that are affected by past treatment and confound the relationship between future treatment and the outcome. In the electoral context, changes in polling for one candidate or the other might affect a candidate's decision to go on the attack, but those attacks might affect future polling.
Typical adjustment for these covariates via regression or matching will lead to post-treatment bias,
while omitting them from these methods will obviously lead to confounding bias. To resolve this dilemma,
Robins and colleagues have developed an inverse probability of treatment weighting (IPTW) approach that allows
for the estimation causal effects of treatment histories without inducing post-treatment bias
\citep{Robins98, Robins98b, Robins99, RobHerBru00}.
This weighting approach is often applied to the estimation of parameters from a model for the marginal mean of the potential outcomes, which these authors call the marginal structural model.

One limitation of the IPTW approach to marginal structural models is that it usually relies on an assumption of
sequential ignorability, which essentially states that there are no unmeasured confounders between the treatment
at time $t$ and the outcome conditional on the treatment and covariate history up to that point.
In the social sciences, this assumption could be suspect when units select into treatment based on data not available to the researcher.
While sensitivity analyses can help diagnose the effects of unmeasured confounding on MSM estimates \citep{BruHerHan04, Blackwell14}, these approaches require knowledge of the severity of the unmeasured confounding, which is difficult for any applied researcher to possess.
When this type of confounding is constant over time and the researcher has panel data on both treatment and outcomes, social scientists often use the so-called linear fixed effects (LFE) approach that
transforms all variables to deviations from their unit means \citep[for a review see][Ch. 5]{AngPis09}.
While this can purge time-constant unmeasured confounding under certain assumptions, the linear fixed effects approach has two important limitations for social science applications.
First, the LFE approach requires the outcome and the treatment to be measured at multiple points in time,
but political outcomes such as final vote share in an election are often single endline measures.
Second, LFEs generally cannot estimate the effects of treatment histories under sequential ignorability
without much stronger assumptions that rule out feedback between the outcome and the treatment \citep{Sobel12, ImaKim19}.

To overcome these issues, this article extends the marginal structural models approach to estimating the effects of time-varying treatments to allow for time-constant unmeasured confounding.
To do so, we propose a straightforward modification to IPTW: to include unit-specific fixed effects in the propensity score model used to construct the inverse-probability weights.
While this approach will lead to an incidental parameters problem for the propensity score model \citep{NeySco48},
we show that if this model is correctly specified and the number of time periods grows at the same rate as the number of units,
the IPTW with fixed effects estimator (IPTW-FE) will lead to a consistent and asymptotically normal estimator for the parameters of the marginal structural model.
This is true even when we only have a single measurement of the outcome after the final instance of treatment. This approach relies on a within-unit version of sequential ignorability, which allows the type of feedback between the treatment and outcome usually ruled out by LFE estimators.
The essential logic of the IPTW-FE is quite simple. If the propensity score model is stable over time and we have a number of time periods,
we can allow for each unit to have a unique offset to the propensity score model that should incorporate any time-constant variables, measured or unmeasured.

In addition to the literature on marginal structural models, we also build on a robust literature on nonlinear panel models.
Initially, these models follow the path of LFEs in side-stepping the incidental parameters problem
by leveraging estimation techniques that did not require direct estimates of the unit-specific unobserved effects \citep[see][for a review]{Lancaster00}.
Unfortunately, these approaches have a severe limitation for our purposes: by avoiding the estimation of the unit-specific effects,
they preclude the ability to calculate the types of predicted probabilities required for IPTW estimation.
Due to this limitation, we build on approaches that focus on large-$T$ approximations, which consider asymptotic results where the number of time periods grows with the number of units. A host of papers have taken this approach to investigate the properties of maximum likelihood estimators of nonlinear panel models with unit (and time) effects \citep{HahNew04, AreHah07, FernandezVal09, HahKue11,  FerWei16, FerWei18}. Many of these approaches have developed bias correction techniques since these estimators are often asymptotically biased. Our approach avoids this issue with these estimators for two reasons. First, we follow the MSM literature and focus on estimating the parameters of the MSM at the slower $\sqrt{N}$ rate rather than the $\sqrt{NT}$ rate so that the asymptotic bias described in this literature converges to 0. Second, we focus on the effect of a finite number of lags of treatment, which limits how much the bias from noisy fixed effect estimation can affect the estimates of the MSM parameters. 

There are some important limitations to our IPTW-FE approach.
First, given the asymptotic setup, this approximation will be more accurate
when the number of units and the number of time periods are not drastically different and when the treatment process provides new information over time. The latter requirement is captured by a strong mixing assumption on the treatment and covariate processes, though we do not impose stationarity.
Second, we derive the asymptotic properties of the IPTW-FE estimator for a MSM that is a function of a fixed number of periods
in the treatment history, rather than the entire history. Third, the performance of IPTW-FE depends heavily on the how severe the unit-specific heterogeneity is.
This is due to extreme unit fixed effects in the propensity score models leading to extreme weights, making estimation less stable.
Relatedly, units that never or always receive treatment in the sample have unit-fixed effects that are not identified.
We explore two procedures for handling these situations: dropping units without treatment variation and imputing their propensity scores
to values close to zero or one.
In spite of these limitations, we find in simulations that IPTW-FE can outperform naive IPTW without fixed effects, even when the unobserved heterogeneity is severe.

Our approach is also related to recent work on causal inference in fixed effects settings. \citet{arkhangelsky2018role} is mostly closely related to our approach here. They investigate how to use inverse probability weighting with fixed effects when a set of sufficient statistics for the treatment process is available, though in a fixed-$T$ setting with no dynamic feedback between the treatment and the outcome and no time-varying covariates. Other work has explained how this dynamic feedback stymies estimation of both contemporaneous effects and the effects of treatment histories with fixed effects assumptions \citep{Sobel12, ImaKim19}. In contrast, our approach allows for feedback between the treatment and the outcome, so long as sequential ignorability holds conditional on the unit-specific effect.

The paper proceeds as follows. We begin in Section~\ref{s:motivation} with a description of out motivating application of the effects of negative advertising in U.S.\ elections. In Section~\ref{s:illustration}, we develop the core intuition for our approach in a simple setting of unit-specific randomized experiments without covariates. Next in Section~\ref{s:review-msm} we review marginal structural models and inverse probability of treatment weighting as they are currently deployed in applied research. When then introduce our fixed-effect approach in Section~\ref{s:iptw-fe}, describing both the assumptions that justify its use and its large-sample properties under these assumptions. In Section~\ref{s:simulation}, we present simulation evidence of the finite-sample performance of this estimator, which shows that it works well, especially when the amount of unmeasured heterogeneity is limited. Finally, we apply the method to the context of negative advertising in Section~\ref{s:application} and conclude with some ideas for future research in Section~\ref{s:conclusion}.

\section{Motivating Application: Negative Advertising in Electoral Campaigns}\label{s:motivation}

In the United States, political advertising is an important way that candidates for office attempt to influence the electorate. One tool that candidates have is the tone of their ads---they might air ads that promote themselves and their agenda, or they may show ads that contrast themselves with their opponent. The latter, which we call negative advertising, is commonplace in campaigns, but its effects are not well understood. In this paper, we apply the IPTW-FE approach to a key empirical question in this literature: what is the effect of negative advertising versus positive advertising on vote shares for a particular electoral candidate? A long literature has explored how and when the tone of an advertisement might affect both the decision to vote (voter mobilization) and the vote choice conditional on voting (persuasion) \citep{LauSigRov07, Blackwell13}, but the results are not conclusive.

Most studies of negative advertising have ignored the sequential nature of this treatment---polling data affects the decision to go negative which in turn affects future polling---though \cite{Blackwell13} used IPTW and MSMs to adjust for this time-varying confounding to investigate the time-varying effects of negative advertising. One worry with that approach, however, is that some campaigns will have baseline higher probabilities of going negative for reasons that could be related to the outcome. For instance, challengers of seats whose incumbents have take unpopular votes or actions might attack more and be more likely to win. Since fully adjusting for these characteristics is difficult, a fixed-effects approach that adjusts for unmeasured baseline confounders could reduce the potential for bias in the estimates of the MSM parameters. Below, we find that our IPTW-FE approach can lead to different substantive conclusions how negative advertising affects electoral outcomes compared to a traditional IPTW approach.

\section{Simple Illustration: Unit-specific Randomized Experiments}\label{s:illustration}

We begin with a simple setting that can illustrate many of the main features of the IPTW-FE methodology. Suppose we observe $N$ units indexed by $i=1,\ldots,N$ each with a binary treatment measured at $T$ points in time, $D_i = \{D_{i1}, \ldots, D_{iT}\}$. We are interested in assessing the effect of the last treatment on some endline measure. We define $Y_i(d_1, \ldots, d_T)$ as the potential outcome under a particular treatment history. For exposition in this simple setting, we assume an extreme no-carryover setting where the outcome only depends on the last period, $Y_i(d_T) \equiv Y_i(d_{i1}, \ldots, d_{i,t-1}, d_T)$.
Below, we will weaken this assumption for marginal structural models in Section~\ref{s:iptw-fe}, but it is useful here to highlight key parts of the approach.  We make the usual consistency assumption, $Y_i = Y_i(1)D_{iT} + Y_i(0)(1 - D_{iT})$. Letting $\tau_d = \E[Y_i(d)]$ for $d = 0, 1$, we define the quantity of interest to be $\tau = \tau_1 - \tau_0$.

We consider a setting where each unit repeatedly makes treatment decisions with unit-specific propensities, but those propensities are unknown to the researcher. In particular, we assume there is a time-constant unit-specific shock, $\alpha_i$, that potentially affects both treatment and the outcome, but that ignorability of treatment assignment holds conditional on that shock:
\[
D_i \indep Y_i(d_T) \mid \alpha_i \qquad \forall i.
\]
We assume that the $D_{it}$ are independently randomly assigned within a unit with probability $\hP(D_{it} = 1 \mid \alpha_i) = \pi(\alpha_i) = \pi_i$, where both $\pi_i$ and $1-\pi_i$ are bounded below by  $c > 0$ for all $i$.

If the unit-specific propensity scores were known, estimation and inference could proceed as usual. Define the following infeasible IPTW estimator:
\begin{equation}
  \label{eq:simple-ipw-infeasible}
  \widetilde{\tau} = \frac{\sum_{i=1}^N (D_{iT} / \pi_i)Y_i}{\sum_{i=1}^N (D_{iT} / \pi_i)} - \frac{\sum_{i=1}^N ((1-D_{iT}) / (1-\pi_i))Y_i}{\sum_{i=1}^N ((1-D_{iT}) / (1-\pi_i))}
\end{equation}
We can verify the consistency of this estimator using the standard IPTW approach. In particular, we can write the first term of the right-hand side of~\ref{eq:simple-ipw-infeasible} as the solution to the sample version of the following population moment condition:
\[
0 = \E\left[\frac{D_{iT}}{\pi_i}\left( Y_i - \tau_1 \right)\right] = \E\left[\frac{D_{iT}}{\pi_i}\left( Y_i(1) - \tau_1 \right)\right] = \E\left[ \frac{\E[D_{iT} \mid \alpha_i]}{\pi_i}\E[Y_i(1) - \tau_1 \mid \alpha_i] \right] = \E[Y_i(1) - \tau_1]
\]
Under regularity conditions, the first term on the right-hand side of~\eqref{eq:simple-ipw-infeasible} will be consistent for $\E[Y_i(1)]$. Applying a similar argument to the second term, we can establish that $\widetilde{\tau}$ is consistent for $\tau$. Using standard asymptotic techniques, we can write this estimator in its asymptotically linear form
\[
\sqrt{N}(\widetilde{\tau} - \tau) = \frac{1}{\sqrt{N}} \sum_{i=1}^N U_i + o_p(1),
\]
where $U_i = U_{i1} - U_{i0}$, and
\[
U_{i1} = (D_{iT}/\pi_i)(Y_i(1) - \tau_1), \qquad  U_{i0} =  ((1-D_{iT})/(1-\pi_i))(Y_i(0) - \tau_0).
\]
The first term in this expansion converges in distribution to $N(0, V)$, where $V = \E[U_i^2]$.


What if, as is almost always the case, we do know the true propensity scores? In the typical cross-sectional case, estimation of the causal effects would not be possible because there could be unmeasured confounding between treatment (captured in the unit-specific propensity scores) and the outcome. In the present setting, however, we have additional information available to us that can adjust for the confounding. In particular, we have the entire treatment history that we can use to estimate the unknown, unit-specific propensity scores. Let $\widehat{\pi}_i = T^{-1} \sum_{t=1}^T D_{it}$ be the sample proportion of treated time periods for unit $i$. We now define the unit-specific IPTW estimator:
\begin{equation}\label{eq:feasible-IPW}
\widehat{\tau} =  \frac{\sum_{i=1}^N (D_{iT} / \widehat{\pi}_i)Y_i}{\sum_{i=1}^N (D_{iT} / \widehat{\pi}_i)} - \frac{\sum_{i=1}^N ((1 - D_{iT}) / (1 - \widehat{\pi}_i))Y_i}{\sum_{i=1}^N ((1 - D_{iT}) / (1 - \widehat{\pi}_i))}
\end{equation}
This plug-in estimator replaces the unknown unit-specific propensity scores with the over-time, within-unit sample averages of the treatment indicators. Somewhat surprisingly, this estimator is consistent and asymptotically normal in spite of the unmeasured confounding, as shown by the following proposition. The proof is in Appendix~\ref{appendix:proof-prop-simple}.

\begin{proposition}\label{prop:simple}
As $N,T \rightarrow \infty$ and $N/T \rightarrow \rho$, where $0 < \rho < \infty$,
the asymptotic distribution of the estimator $\widehat{\tau}$ defined in Equation~\eqref{eq:feasible-IPW}
is given by $\sqrt{N}(\widehat{\tau} - \tau) \overset{d}{\to} \mathcal{N}(0, V)$, where $V = \E[U_i^2]$.
\end{proposition}
Proposition~\ref{prop:simple} shows that when the sample size does not dominate the number of time periods and when the history of treatment provides independent information about the propensity score of interest, using the unit-specific propensity scores in an IPW estimator leads to a consistent and asymptotically normal estimator with variance the same as if the propensity scores were known. This is in spite of the fact that there may be unmeasured confounding between treatment and the outcome across units. Furthermore, unlike most extant ways of adjusting for time-constant unmeasured confounding, this result only requires repeated measurements of the treatment, not of both treatment and the outcome.

How does this approach avoid the incidental parameters problem? It is instructive to review a Taylor expansion of this estimator around the true value of propensity scores. Letting $\widehat{\tau}_1$ be the first term in~\ref{eq:feasible-IPW}, we can expand it as
\[
\sqrt{N}(\widehat{\tau}_1 - \tau_1) = \frac{1}{\sqrt{N}} \sum_{i=1}^N U_{i1} - \frac{1}{\sqrt{N}} \sum_{i=1}^N \frac{(\widehat{\pi}_i - \pi_i)}{\pi_i} U_{i1} + o_p(1/T)
\]
The first term on the left-hand side of this expansion is the treatment group's contribution to the influence function when the propensity scores are known, and the second term is the first-order effect of estimating the propensity scores. Because of the within-unit ignorability and independence of treatment over time conditional on $\alpha_i$, we can show that the second term is $O_p(1/\sqrt{T})$, so it can be ignored in a first-order asymptotic approximation. In a typical panel data setting, there would be an outcome in every time period and typical methods require a $\sqrt{NT}$ rate of convergence, in which case the second term would contribute bias and variance to the asymptotic distribution. While we have focused on the effect of the last blip of treatment, it is straightforward to extend this result to investigate the effect of some subset of the treatment history so long as the length of the subset is fixed as $N,T \rightarrow \infty$. The fixed window ensures that the incidental parameters bias does not dominate the asymptotic distribution.

One advantage of the unit-specific weights is that they act as \emph{balancing weights} for any baseline covariate, measured or unmeasured. In particular, the true propensity scores,  $\pi_i = \hP(D_{i}=1 \mid \mathcal{F}_i)$, where $\mathcal{F}_i$ is the set of all possible time-constant variables influencing treatment. That is, for any $Z_i \in \mathcal{F}_i$, $\E[D_{iT}Z_i/\pi_i] = \E[Z_i]$ and
\[
\E\left[\frac{D_{iT}Z_i}{\pi_i} - \frac{(1 - D_{iT})Z_i}{1-\pi_i} \right] = 0,
\]
so that the true weights balance the baseline covariates on average.

This setting was highly restricted to demonstrate the core intuition behind the IPTW-FE approach. There were no covariates, time-varying or otherwise, and we relied on a nonparametric estimator for the propensity score. In the rest of the paper, we expand our scope to focus on the more general class of marginal structural models.

\section{Review of Marginal Structural Models and Inverse Propensity Score Weighting}\label{s:review-msm}

The combination of marginal structural models and inverse probability of treatment weighting was developed by \cite{Robins98b} and has since become an important method across a number of scientific domains. \cite{RobHerBru00} provides a general introduction to the method. A robust methodological literature has built up around the method, focusing on stabilizing the construction of the weights \citep{ColHer08, XiaMooAbr13, ImaRat15, KalSan19}, using machine learning methods to make estimation more flexible \citep{Diavan11, BruLogJar15}, or developing doubly robust versions of the approach \citep{BanRob05, RotLeiSue12}. Our contribution to this literature is to show how these methods may be applied when a researcher suspects there may be time-constant unmeasured confounding.

We now review the use of IPTW in the context of marginal structural models. Two main features distinguish it from the previous discussion. First, we allow for potential outcomes to be a function of treatment history. Let $\overline{D}_{it} = (D_{i1},\ldots,D_{it})$ be the treatment history up to time $t$ and $\underline{D}_{it} = \{D_{it}, \ldots, D_{iT}\}$ be the history from $t$ to $T$.  Let $\overline{D}_i = \overline{D}_{iT}$, where these take values in $\mathcal{D}_T \in \{0,1\}^T$. We define the potential outcomes as  $Y_i(\overline{d})$, where $\overline{d} \in \mathcal{D}_T$, which is the outcome that unit $i$ would have if they had followed treatment history $\overline{d}$. Second, we allow for the possibility of time-varying confounders. Let $X_{it}$ be a vector of time-varying covariates that are causally prior to $D_{it}$. Even though this vector is labelled in period $t$, we assume it can include lags of contemporaneous covariates. We define $\overline{X}_{it}$, $\underline{X}_{it}$ and $\underline{X}_i$ similarly to the treatment history.

The MSM methodology is based on a sequential ignorability assumption that treatment at time $t$ is unrelated to the potential outcomes conditional on (some function of) the history of treatment and the time-varying covariates. In particular, there is some vector of time-varying covariates, such that,
\[
Y_i(\overline{d}) \indep D_{it} \mid X_{it}, \overline{D}_{i,t-1},
\]
where again $X_{it}$ may include information about the history of the covariates as well. This assumption is a time-varying version of a selection-on-observeables assumption applied repeatedly to treatment in each period.  One drawback with this approach in the social sciences is that units may have differing baseline probabilities of treatment based on traits that are difficult to measure, as we had in the last section. In the context of negativity, this may occurs if a candidate is facing strong challenger that is not fully captured in polling data. This limitation of sequential ignorability is one motivation for developing the fixed-effects approach we introduce below.

A marginal structural model is a model for the marginal mean of the potential outcomes as a function of the treatment history:
\[
\E[Y_i(\overline{d})] = g(\overline{d}; \gamma_0)
\]
The dimensionality of $\overline{d}$ grows quickly in $T$, so even when $T$ is moderate, $g(\cdot)$ will usually impose some parametric restrictions on the response surface. Even if these modeling restrictions are correct, the observed conditional expectation function $\E[Y_i\mid \overline{D}_i = \overline{d}] \neq g(\overline{d}; \gamma_0)$ due to confounding by $X_{it}$.  On the other hand, including the covariates in the conditional expectation will lead to  post-treatment bias so that $\E[Y_i \mid \overline{D}_i = \overline{d},\overline{X}_i] \neq g(\overline{d}; \gamma_0)$. \citet{Robins99} showed how an inverse probability of treatment weighting scheme could avoid these two biases. In particular, he showed that a weighted conditional expectation can recover the parameters of the MSM when the weights are proportional to the inverse of the conditional probability of a the unit's treatment history given their covariate history. Let $\pi_t(\overline{d}_{t-1}, x) = \hP(D_{it} = 1 \mid \overline{D}_{i,t-1} = \overline{d}_{t-1}, X_{it} = x)$ and let $\pi_{it} = \pi_t(\overline{D}_{i,t-1}, X_{it})$. Then, the IPTW weights for our MSM become
\begin{equation}\label{eq:msm-weights}
W_i = \prod_{t=1}^T \pi_{it}^{-D_{it}}(1 - \pi_{it})^{-(1-D_{it})}
\end{equation}
With these weights, \citet{Robins99} showed that $\E[\bm{1}\{\overline{D}_i = \overline{d}\} W_{i} Y_i] = g(\overline{d}; \gamma_0)$.

In observational studies, the propensity scores used to construct the weights are not usually known to the analyst and so must be estimated. The standard approach to this in the MSM literature is to specify a parametric model for treatment and estimate its parameters via maximum likelihood. Define the a parametrization of the propensity score  $\pi(x, \overline{d}_t; \beta)$, where we define the true value of this parameter as $\pi(x, \overline{d}_t; \beta_0) = \hP(D_{it} = 1 \mid X_{it} = x, \overline{D}_{it} = \overline{d}_t)$. We then define the estimated propensity scores as $\widehat{\pi}_{it} = \pi(X_{it}, \overline{D}_{it};\widehat{\beta})$, where  $\widehat{\beta}$ is the MLE:
\[
  \begin{gathered}
    \widehat{\beta} = \argmax_{\beta} \frac{1}{NT} \sum_{i=1}^N \sum_{t=1}^{T} \ell_{it}(\beta),\\
    \ell_{it}(\beta) = D_{it}\log \pi(X_{it}, \overline{D}_{i,t-1}; \beta) + (1 - D_{it})\log(1 - \pi(X_{it}, \overline{D}_{i,t-1}; \beta)).
  \end{gathered}
\]
These estimated propensity scores can then be used to generate estimated weights, $\widehat{W}_i = \prod_{t=1}^T \widehat{\pi}_{it}^{-D_{it}}(1 - \widehat{\pi}_{it})^{-(1-D_{it})}$. With these estimated weights, an IPTW estimator for the MSM can be constructed by solving the empirical version of the following estimating equation for $\gamma$:
\[
\E\left\{\widehat{W}_{i} h(\underline{D}_i)(Y_i - g(\underline{D}_i;\gamma)) \right\} = 0.
\]
For example, when $g(\cdot)$ is the identity function, then this approach reduces to weighted least squares. \cite{Robins98} established this procedure as producing a consistent and asymptotically normal estimator for the parameters of the MSM.

The weights in~\eqref{eq:msm-weights} can often be unstable when the true or estimated propensity scores are close to one or zero, which can lead to highly variable estimates. A common practice in this case is to include a stabilizing numerator that is the marginal probability of the treatment history, $\overline{\pi}_{it} = \hP(D_{it} = 1 \mid \overline{D}_{i,t-1})$. In this case, the stabilized weights become
\[
\widetilde{W}_i = \prod_{t=1}^T \left(\frac{\overline{\pi}_{it}}{\pi_{it}}\right)^{D_{it}}\left( \frac{1 - \overline{\pi}_{it}}{1 - \pi_{it}} \right)^{1-D_{it}}.
\]
Another common practice is to trim the weights to additionally guard against unstable causal parameter estimates \citep{ColHer08}.

\section{Fixed-effect Propensity Score Estimators}\label{s:iptw-fe}

\subsection{Setting and Assumptions}

We now focus on estimating propensity scores with fixed effects for MSMs when time-constant unmeasured confounding exists. As with the traditional MSM case, we assume that  $(Y_i, \overline{D}_{i}, \overline{X}_i)$ are independent across observations. In order to adjust for unit-specific heterogeneity, we do require restrictions beyond the typical MSM case. First and foremost, we focus on  marginal structural models for a treatment history of a fixed length rather than the entire treatment history. In particular, we focus on MSMs of the form $\E[Y_{i}(\underline{d}_{T-k})] = g(\underline{d}_{T-k};\gamma)$, where $\underline{d}_{T-k} = (d_{T-k},\ldots, d_T)$, $k$ is fixed, and the parameter vector $\gamma$ is of length $J$. This restriction is a choice of the quantity of interest, not a substantive assumption about the effect of the treatment history. By the usual consistency assumption, we can define these ``shorter'' potential outcomes as $Y_i(\underline{d}_{T-k}) \equiv Y_i(\overline{D}_{i,T-k-1}, \underline{d}_k)$,
so that treatment history before $k$ lags acts more like a baseline confounder. Compared to typical MSM practice, the main limitation of this restriction is to rule out functional forms where the cumulative sum of the treatment history is included as part of the MSM.

We now describe the key identification assumption of the IPTW-FE approach, which is a combination of the unit-specific randomized experiments assumption of Section~\ref{s:illustration} and the standard MSM framework in Section~\ref{s:review-msm}. Let $\underline{X}_{i,t+1}(d_t)$ and $\underline{D}_{i,t+1}(d_t)$ represent the potential outcomes of the future covariate and treatment histories for an intervention on the treatment process at time $t$.
\begin{assumption}[Unit-specific Sequential Ignorability]\label{a:sequential-ignorability} For all $i$ and $t$,
\[\{Y_i(\overline{d}), \underline{X}_{i,t+1}(\overline{d})\} \indep D_{it} \mid \overline{X}_{it}, \overline{D}_{i,t-1} = \overline{d}_{t-1}, \alpha_i.\]
\end{assumption}
\noindent Assumption~\ref{a:sequential-ignorability} states that conditional on the unit-specific effect, the treatment history, and (a function of) the covariate history, treatment is independent of future potential outcomes for both the outcome and the covariate process. In essence, treatment is randomized with respect to future covariates and the outcome, conditional on the past and time-constant features of the unit. This assumption allows for both time-varying confounding by measured covariates and time-constant confounding by measured and unmeasured covariates. We do assume that the time-constant unmeasured confounding can be captured by the unidimensional, $\alpha_i$, which might represent a combination of several unit-specific factors.

Assumption~\ref{a:sequential-ignorability} involves potential outcomes histories of various lengths, $Y_i(\underline{d}_t)$, but above we defined our main marginal structural models in terms of treatment histories of fixed length, $\E[Y_i(\underline{d}_{T-k})]$. Thus, the requirements of sequential ignorability go beyond the treatments of interest in the marginal structural model and apply to the potential outcomes for the entire treatment history. This allows for the fixed-effect propensity score estimators to be consistent even without a no-carryover assumption that would assume that treatment before $T-k$ has no effect on the outcome.

We now turn to defining the propensity scores with fixed effects. We assume a parametric model for the this propensity score (up to the unmeasured heterogeneity) with a function of the treatment history acting as the covariates: $V_{it} = b(\overline{X}_{it}, \overline{D}_{i,t-1})$ where $V_{it}$ is a $k \times 1$ vector. Let $\hP(D_{it} = 1 \mid \overline{X}_{it} = \overline{x}_{t}, \overline{D}_{i,t-1} = \overline{d}_{t-1}, \alpha_i) = \hP(D_{it} = 1 \mid V_{it} = v, \alpha_i) =  \pi_{it}(v; \beta, \alpha_i)$ be the probability of treatment given the covariate and treatment histories, where $v = b(\overline{x}_{t}, \overline{d}_{t-1})$. We define $\pi_{it}(\beta, \alpha_i) = \pi(V_{it}; \beta, \alpha_i)$ and assume the following functional form:
\[
\pi(v; \beta, \alpha_i) = F(v^{\top}\beta + \alpha_{i}),
\]
where $F(\cdot)$ is the cumulative distribution function of either the standard normal or standard logistic distribution. One restriction here is that the function $b(\cdot)$ is time-constant which allows us to pool information across the time dimension effectively. Let $\alpha_{0} = (\alpha_{10}, \ldots, \alpha_{N0})$ and $\beta_0$ be the values of the parameters that generate the treatment process. In particular, we assume that these values are the solution to the following population conditional maximum likelihood condition
\begin{equation}
  \label{eq:pop-mle}
(\beta_0, \alpha_0) = \argmax_{(\beta, \alpha) \in \mathbb{R}^{d_\beta + N}} \sum_{i=1}^N\sum_{t=1}^T \E[\ell_{it}(\beta, \alpha) \mid \alpha_i],
\end{equation}
where the expectation is with respect to the distribution of the data conditional on the unobserved effect. Under Assumption~\ref{a:panel-data} below, these parameters will be identified from the model. This explicitly assumes a correctly-specified parametric model for the propensity scores, though this assumption is common in applications of MSMs and the current assumption is still considerably weaker than those settings since it allows for unit-specific heterogeneity.

To account for the time-constant unmeasured confounding, we construct weights with these unit-specific effects to  estimate the  MSMs. In particular, we use the following weights
\[
W_{i}(\beta, \alpha_i) = \prod_{j=1}^k \left( \frac{1}{\pi_{i,T-j}(\beta, \alpha_i)} \right)^{D_{i,T-j}} \left( \frac{1}{1-\pi_{i,T-j}(\beta, \alpha_i)} \right)^{1 - D_{i,T-j}},
\]
where we only take the product over the last $k$ time periods because are quantities of interest focuses on those periods. As with the standard MSM case, we can replace the numerator with the marginal probability of the treatment history, $\overline{\pi}_{it}$, which can stabilize the variance of the estimator without affecting identification.

The IPTW approach to estimating this MSM is to rely on the estimating equation
\[
U_i(\gamma, \beta, \alpha_i) = W_i(\beta, \alpha_{i})h(\underline{D}_{i,T-k})(Y_i - g(\underline{D}_{i,T-k}; \gamma)),
\]
where $h(\cdot)$ is a function with $J$-length output, chosen by the researcher.  For example, if $Y_i$ is continuous and $g$ is linear and additive, it is common to use $h(\underline{D}_{i,T-k}) = \underline{D}'_{i,T-k}$. Under the fixed-effects sequential ignorability assumption and the MSM has the follow restriction:
\begin{equation}\label{eq:msm-population-condition}
    \E\left[ U_i(\gamma_0, \beta_0, \alpha_{i0})\right] = 0
  \end{equation}
This is an identification result because the restriction identifies the causal parameters, $\gamma_0$, solely in terms of sample quantities (up to the propensity score parameters). This result follows the standard g-computation algorithm with the unit-specific heterogeneity, $\alpha_i$, included in the place of a baseline covariate \citep{Robins99, Robins00}.

Of course, to build a feasible estimator for $\gamma$, we need estimates of the propensity score parameters. We will consider first-step estimators that find the MLEs for these parameters. To allow for fixed effects in these models, we make the following assumptions.

\begin{assumption}[Treatment Regularity Conditions]\label{a:panel-data}
Let $\nu > 0$, $\mu > 4(8 + \nu)/\nu$, and $\mathcal{B}_{0}(\epsilon)$ is an $\epsilon$-neighborhood of $(\beta_0, \alpha_{i0})$ for all $i,t,N,T$.
\begin{enumerate}[label=(\roman*), ref=\theassumption(\roman*)]
\item \label{a:rates} (Asymptotics) Let $N,T \rightarrow \infty$ such that $N/T \rightarrow \rho$ where $0 < \rho < \infty$.
\item \label{a:sampling} (Sampling) For all $N$ and $T$,  $\{(Y_i(\underline{d}), \underline{D}_i, \underline{X}_i, \alpha_i): i = 1,\ldots,N\}$ are i.i.d. across $i$. Letting $Z_{it} = (D_{it}, X_{it})$ for $t = 1,\ldots,T$ and $Z_{i,T+1} = (Y_i(\underline{d}))$, then for each $i$, $\{Z_{it}: t = 1,\ldots,T+1\}$ is $\alpha$-mixing conditional on $\alpha_i$ with mixing coefficients satisfying $\sup_i a_i(m) = O(m^{-\mu})$ as $m\rightarrow \infty$ where
  \[
    a_i(m) \equiv \sup_t \sup_{A\in\mathcal{A}_{it}, B \in \mathcal{B}_{i,t+m}} |\hP(A \cap B) - \hP(A)\hP(B)|,
  \]
  and $\mathcal{A}_{it}$ is the sigma field generated by $(Z_{it}, Z_{i,t-1}, \ldots)$ and $\mathcal{B}_{i,t}$ is the sigma field generated by $(Z_{it}, Z_{i,t+1}, \ldots)$.
  \item \label{a:bounded-derivatives} We assume that $(\beta, \alpha) \mapsto \ell_{it}(\beta, \alpha)$ is four-times continuously differentiable over $\mathcal{B}_{0}(\epsilon)$ almost surely.  The partial derivatives of $\ell_{it}(\beta, \alpha)$ with respect to the elements of $(\beta, \alpha)$  are  bounded in absolute value uniformly over $(\beta,\alpha) \in \mathcal{B}_{0}(\epsilon)$ by a function $M(Z_{it}) > 0$ almost surely and $\max_{i,t} \E[M(Z_{it})^{8+\nu}]$ is almost surely uniformly bounded over $N,T$.
  \item \label{a:bounded-propensity-scores} For all $i,t$ we have $\pi_{it}(\beta,\alpha)$ bounded away from $0$ and $1$ uniformly over $(\beta, \alpha) \in \mathcal{B}_0(\epsilon)$.
  \item \label{a:concavity} (Concavity) For all $N$, $T$ $(\beta, \alpha) \mapsto \ell_{it}(\beta, \alpha)$ is strictly concave over $\mathbb{R}^{\text{dim}(\beta)+1}$ almost surely.  Furthermore, there exists $b_{\min}$ and $b_{\max}$ such that for all $(\beta, \alpha) \in \mathcal{B}_0(\epsilon)$, $0 < b_{\min} \leq -\E[V_{it\alpha} \mid \alpha_i] \leq b_{\max}$ almost surely uniformly over $i$, $t$, $N$, and $T$.
  \end{enumerate}
\end{assumption}

Assumption~\ref{a:panel-data} mostly derives from \citet{FerWei16}, who used them to establish the asymptotic properties of nonlinear panel models with unit- and time-specific effects, though we focus only on unit effects. Assumption~\ref{a:rates} establishes the large-N, large-T asymptotic framework, which has been widely used for nonlinear panel models in econometrics \citep{HahNew04, AreHah07, FernandezVal09, HahKue11,  FerWei16, FerWei18}. The strong mixing process in Assumption~\ref{a:sampling}, which allows us to rely on laws of large numbers and central limit theorems in the time dimension. This assumption is substantially weaker than independence over time or even stationarity. In particular,  it allows for time trends which are a common feature of propensity score models in MSMs. The i.i.d. nature of the distribution of the data and the fixed effects across units is common to IPTW approaches and allows us to take averages over the unit-specific heterogeneity and has been used before for average partial effects in nonlinear panel models \citep{FerWei16}. It is possible to replace this assumption with stationarity of $X_{it}$ over time, but this would rule out lagged treatment in the propensity score model along with time trends.

Assumption~\ref{a:bounded-derivatives} requires the log-likelihood of the propensity score model and its derivatives to be sufficiently smooth to allow for the higher-order asymptotic expansions we use. With a binary response, this assumption could be replaced by a moment condition on the distribution of the covariates. We invoke a locally uniform version of positivity in Assumption~\ref{a:bounded-propensity-scores}. Note that Assumption~\ref{a:bounded-propensity-scores} implicitly restricts $\alpha_i$, since if $\alpha_i$ were completely unrestricted, then we may have $\pi_{it} \rightarrow \infty$. Finally, Assumption~\ref{a:concavity} ensures that the MLE is identified and should be satisfied in the usual parametric models used for binary data when the covariates, $V_{it}$ vary in the time and unit dimensions.

In addition to these assumptions on the treatment process, we also make the following regularity conditions on the marginal structural model and outcome.

\begin{assumption}[Outcome Regularity Conditions]\label{a:msm-assumptions}
  Let $\nu > 0$ and $\mathcal{B}_{0}(\epsilon)$ is an $\epsilon$-neighborhood of $(\gamma_0, \beta_0, \alpha_{i0})$ for all $i, N$.
  \begin{enumerate}[label=(\roman*), ref=\theassumption(\roman*)]
  \item \label{a:bounded-moments} (Bounded outcome moments) $\E[|Y_{i}(d)|^{4+\nu}]$ and $\E[|Y_{i}(d)|^{4+\nu} \mid \alpha_i, \underline{D}_i, X_{iT}]$ are bounded by finite constants, uniformly over $i$.
  \item \label{a:msm-regularity} (MSM regularity) The parameters $\phi = (\gamma, \beta, \alpha) \in \interior \Phi$ where $\Phi$ is a compact, convex subset of $\mathbb{R}^{J+R+1}$ with $J = \dim(\gamma)$ and $R = \dim(\beta)$. The map $\gamma \mapsto U_i(\gamma, \beta, \alpha)$ is continuously differentiable over $(\gamma, \beta, \alpha) \in \mathcal{B}_{0}(\epsilon)$ with $\E[\sup_{\gamma \in \mathcal{B}_{0}(\epsilon)} \Vert\partial_{\gamma} U_i(\gamma, \beta, \alpha)\Vert] < \infty$.
 \end{enumerate}
\end{assumption}

Assumption~\ref{a:bounded-moments} ensures the potential outcomes have sufficiently bounded (conditional) moments. Assumptions~\ref{a:msm-regularity} is a set of standard regularity conditions for the marginal structural model.

\subsection{Proposed Method}

We propose a two-step approach to estimating the parameters of the marginal structural model using inverse probability of treatment weighting. These two steps are:
\begin{enumerate}
\item Estimate the parameters of the propensity score model $(\widehat{\beta}, \widehat{\alpha}_i)$ using conditional maximum likelihood treating the unit-specific effects $\alpha_i$ as fixed parameters to be estimated. Construct estimated weight $W_i(\underline{D}_{i,T-k}; \widehat{\beta}, \widehat{\alpha}_i)$.
\item Pass the estimated weights to a weighted estimating equation $N^{-1}\sum_{i=1}^NU_i(\widehat{\gamma}, \widehat{\beta}, \widehat{\alpha}_i) = 0$ to obtain estimates of the MSM parameters, $\gamma$.
\end{enumerate}

The first step in this procedure can be implemented with a sample conditional maximum likelihood estimator. Letting $\widehat{\alpha} = (\widehat{\alpha}_1, \ldots, \widehat{\alpha}_N)$, we have
\begin{equation}
  \label{eq:mle}
  (\widehat{\beta}, \widehat{\alpha}) = \argmax_{(\beta, \alpha) \in \mathbb{R}^{d_{\beta} + N}} \E_{NT}[\ell_{it}(\beta, \alpha_i)]
\end{equation}
Under these assumptions we use the following maximum likelihood estimators:
\[
  \widehat{\beta} = \argmax_{\beta} \frac{1}{NT} \sum_{i=1}^N \sum_{t=1}^{T} \ell_{it}(\beta, \widehat{\alpha}_i(\beta)), \qquad
  \widehat{\alpha}_i(\beta) = \argmax_{\alpha} \frac{1}{T} \sum_{t=1}^T \ell_{it}(\beta, \alpha).
\]
These maximum likelihood estimates are subject to the usual incidental parameters problem that results in bias that shrinks as $T\rightarrow\infty$. Even when $N$ and $T$ grow at the same rate, \citet{HahNew04} showed that these types of MLE estimators are not $\sqrt{NT}$-consistent and a large literature has developed proposing several bias correction techniques \citep{AreHah07, FerWei18}. We sidestep these issues in our own results because we target the slower convergence rate of $\sqrt{N}$ that is typical in the MSM literature.

To obtain estimates of the MSM parameters, $\widehat{\gamma}$, we use the sample version of the MSM moment condition~\eqref{eq:msm-population-condition},
\[
\frac{1}{N} \sum_{i=1}^N W_i(\underline{D}_{i,T-k};\widehat{\beta}, \widehat{\alpha}_i)h(\underline{D}_{i,T-k})(Y_i - g(\underline{D}_{i,T-k};\widehat{\gamma})) = \frac{1}{N} \sum_{i=1}^N U_i(\widehat{\gamma}, \widehat{\beta}, \widehat{\alpha}_i) = 0.
\]
This estimator depends on the link function for the marginal structural model and a function $h(\cdot)$. One particularly straightforward estimator in this class is  weighted least squares for the identity link with continuous outcomes. Often $h(\cdot)$ can be chosen to enhance the efficiency of the estimator \citep{Robins99}, but we do not explore that here.
We now show in Theorem~\ref{thm:msm-iptw-fe} that under regularity conditions and the above assumptions, this estimator is consistency and asymptotically normal.
The proof is in Appendix~\ref{appendix:proof-thm-iptwfe}.

\begin{theorem}\label{thm:msm-iptw-fe}
Under Assumptions~\ref{a:sequential-ignorability}, \ref{a:panel-data}, and~\ref{a:msm-assumptions}, $\widehat{\gamma} \inprob \gamma_0$ and
\begin{equation}
\sqrt{N}(\widehat{\gamma} - \gamma_0) \xrightarrow{d} N(0, V_{\gamma_0}),
\end{equation}
where $V_{\gamma_0} = G^{-1}\E[U_iU_i^{\top}]G^{-1}$ and
\[
G = \E\left\{ \frac{\partial U_i(\gamma, \beta, \alpha)}{\partial \gamma} \right\}_{\gamma=\gamma_0} \qquad U_i = U_i(\gamma_0, \beta_0, \alpha_{i0}).
\]
\end{theorem}
We can build a consistent variance estimator in the usual way with $\widehat{V}_{\gamma} = \widehat{G}^{-1}\widehat{\Omega}\widehat{G}^{-1}$, where
\[
\widehat{G} = \frac{1}{N} \sum_{i=1}^N \frac{\partial \widehat{U}_i}{\partial \gamma}, \qquad \widehat{\Omega} = \frac{1}{N} \sum_{i=1}^N \widehat{U}_i\widehat{U}_i^{\top}, \qquad \widehat{U}_i = U_i(\widehat{\gamma}, \widehat{\beta}, \widehat{\alpha}_i).
\]
This is a standard sandwich estimator for estimators based on estimating equations.

Theorem~\ref{thm:msm-iptw-fe} establishes that the IPTW-FE for MSMs is asymptotically normal and that we can asymptotically ignore the estimation of the weights. In the standard IPTW case, the estimation of the weights does impact the distribution of the MSM estimates. Here, however, the estimation of the weights doesn't affect the first-order asymptotic distribution because we are using $NT$ observations to estimate the propensity score parameters but only using a fraction of the observations, $Nk$, to create the weights, where $k$ is fixed as $T\to\infty$. Thus, the $\widehat{\beta}$ converges much faster than $\widehat{\gamma}$ and so we can ignore its estimation.

In typical nonlinear panel models, plugging in noisy estimates of the fixed-effect parameters leads to a bias that converges to 0 slowly enough to create asymptotic bias. In our setting, however, the strong mixing property of the treatment process that this bias fades over time and so allows us to ignore the estimation of the fixed-effect parameters as well. In the literature on nonlinear panel models, there is a similar result for estimating partial effects, or differences in the conditional expectation, as opposed to parameters of the nonlinear model. For example, \cite{FerWei18} showed how these average partial effects can converge at a slower rate with parameter estimation not having a first-order effect on the asymptotic distribution \cite[see also][]{FerWei16}. The current approach is similar since we are only interested in the parameters of the weighting model insofar as they provide consistent estimates of the IPTW weights.

This result establishes that it is possible to adjust for unmeasured baseline confounding in MSMs when the time dimension is long and provides sufficiently new information within units. The quality of this adjustment will depend on both how long the panels are and how severe the unmeasured heterogeneity is. A second-order expansion of the estimator shows that second-order bias (which can be ignored in our asymptotic analysis) is inversely related to the propensity scores. Thus, strong unit-specific heterogeneity will push propensity scores close to zero or one and create more finite-sample bias. Longer panels helps with this finite-sample bias since these second-order terms will be of order $O_P(1/\sqrt{T})$. A fruitful avenue for future research would be to use analytic or computational approaches like the jackknife to adjust for these second-order terms.

\subsection{Trimming Weights}

One drawback of the IPTW-FE approach is that the fixed-effect parameters of the propensity score model are not identified when units are either always treated or always control. Even if when we maintain the population-level positivity assumption, this in-sample positivity violation means that some units will have undefined weights. We propose three ways to address this issue. First, one could simply omit the no-treatment-variance units and estimate the parameters of the MSM for the units that have at least one treated or control period. This is the simplest procedure but could induce confounding bias, especially if the $\alpha_i$ have a nonlinear relationship with the outcome. Second, we could use an ad hoc rule for imputing propensity scores of the no-treatment-variance units.
For example, we could set these units to have $\widehat{\pi}_{it} = 0.01$ if $D_{it} = 0$ for all $t$ and $\widehat{\pi}_{it} = 0.99$ if $D_{it} = 1$ for all $t$. Depending on the lag length $k$ in the MSM and the exact trimming this may lead to extreme weights which themselves could require trimming. Alternatively, one could place bounds on the range of the unit-specific effects in the MLE estimation to $\alpha_i \in [a_0, a_1]$ and set the estimates of those effects as $\widehat{\alpha}_i = a_0$ or $\widehat{\alpha}_i = a_1$ if $D_{it} = 0$ or $D_{it} = 1$ for all $t$, respectively. The amount of trimming of the weights in this approach amounts to a bias-variance trade-off similar to weight trimming in standard IPTW estimators for MSMs \citep{ColHer08}.

Finally, one alternative approach to handling positivity violations would be to focus on a different quantity of interest. \cite{Kennedy19} proposed estimating the effect of incremental propensity score interventions, which are interventions that shift the propensity score rather than set treatment histories to specific values. The identification and estimation of these effects does not depend on positivity, and under the assumption of a correctly specified propensity score model, a simple inverse probability weighting estimator is available \citep[p. 650]{Kennedy19}.

\section{Simulation Evidence}\label{s:simulation}

\input{psfe_simulation.tex}

\section{Empirical Application: The Effectiveness of Negative Advertising}\label{s:application}

We now apply these techniques to estimate the effectiveness of negative advertising in U.S.\ Senate and Gubernatorial elections. We build on \citet{Blackwell13} who investigated the same question using an MSM approach without fixed effects for elections over the period from 2000 to 2008. We expand the data to include additional Senate races from 2010 until 2016, and focus on the effect of Democratic candidate negativity on three outcomes: the Democratic percentage of the two-party vote, percent of the voting-eligible population casting Democratic votes, and percentage of the voting-eligible population casting Republican votes. The latter two outcomes use the voting-eligible population as a denominator, allow us to explore the possibility that Democratic negativity mobilizes each party differently. We organize the data into a race-week panel and focus on the time period between the primary election for the race and the general election, so that we have $N=201$ with the length of the campaign ranging from 8 to 40 weeks with a median of 20 weeks. Our marginal structural model is:
\[
\E[Y_i(\underline{d})] = \gamma_0 + \gamma_1\left(\sum_{k=0}^4 d_{T-k}\right),
\]
where the time index here is weeks of the campaign. The main quantity of interest, $\gamma_1$, can be interpreted as the effect of additional week of negative advertising in the last five week on the outcomes. The data on advertising comes from the Wesleyan Media Project, which provides data on the number, timing, and tone of political television ads for most recent elections. The outcomes come from the MIT Election Data Science Project, and we collected polling data from several polling aggregators.

We apply several different estimation approaches to this MSM: the proposed IPTW-FE approach, a standard IPTW approach without fixed effects, and a naive approach that ignores time-varying covariates all together. For each of these methods, we consider ``vanilla'' MSMs that only include treatment history and augmented MSMs that additional control for baseline covariates. For the weighting model, we included various time-varying covariates: average Democratic share of the two-party preferences in polls in the previous week (and the square of this term), the average percentage reporting undecided or voting for third-party candidates in the previous week, measures of Republican negativity over the last six weeks, the cumulative number of ads shown by the Democrat and Republican (and their squared terms). For the fixed effects approach, we additionally include a race fixed effect term in the specification. For the IPTW approach without fixed effects, we include several baseline covariates including the length of the general election campaign,  baseline polling, whether the Democrat is the incumbent, and an indicator for the type of elected position.

\begin{figure}[t!]
  \centerline{\includegraphics[scale=0.75]{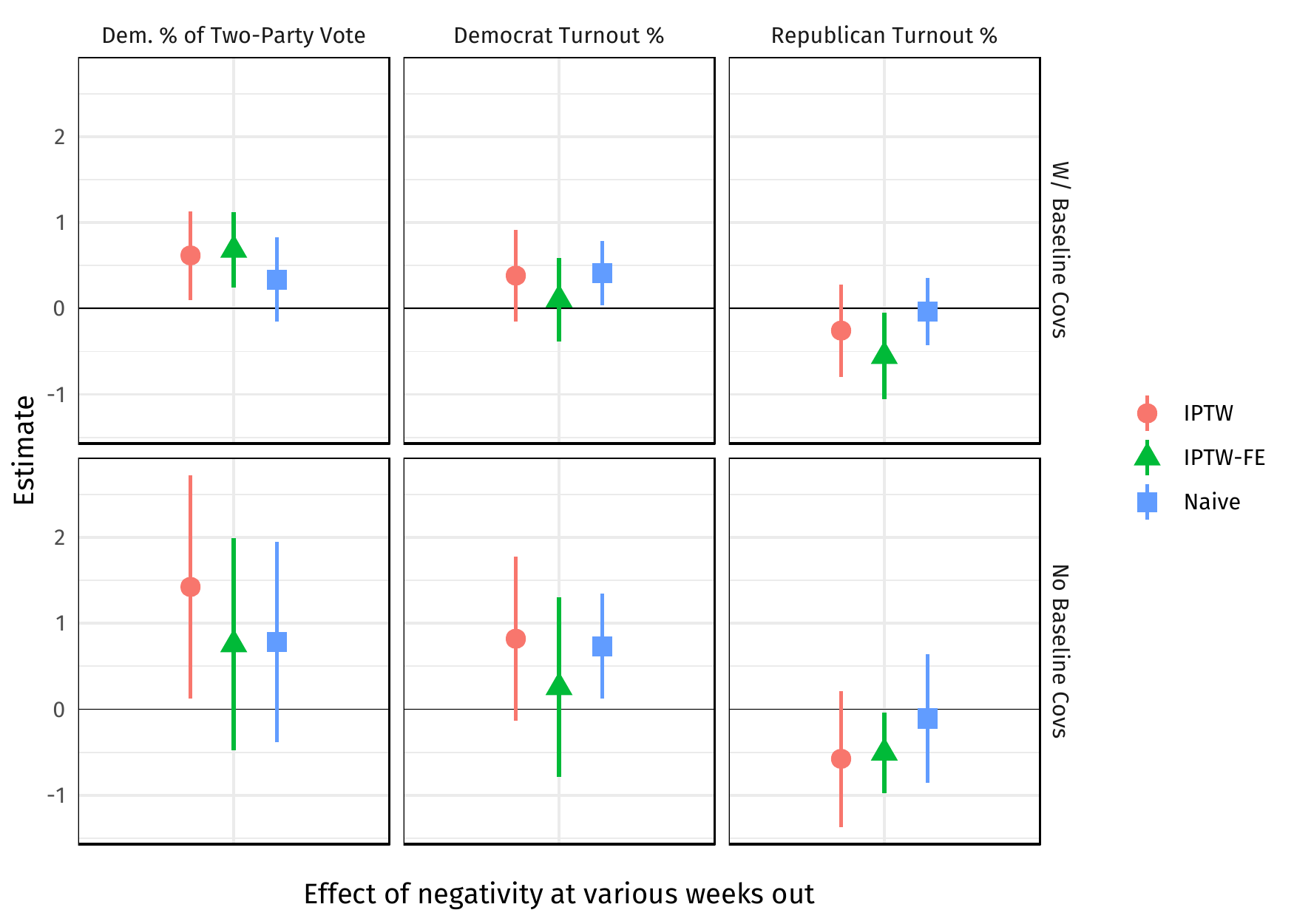}}
  \caption{Estimated effects of the number of weeks of negativity in the last five weeks of the campaign with different methods.}\label{fig:negativity-coefplot}
\end{figure}

Figure~\ref{fig:negativity-coefplot} shows the results of these methods for each of the outcomes. Substantively, the methods generally agree that there is a positive effect of Democratic negativity on Democratic electoral performance, though there is disagreement between the methods on why. The results of the standard IPTW approach indicates that this electoral edge comes from a positive effect on Democratic turnout, with relatively little impact on Republican turnout. The IPTW-FE approach, on the other hand, indicates that negativity actually demobilized Republican turnout without a large effect on Democratic turnout, raising the overall Democratic share of the vote. One other key difference between the IPTW and IPTW-FE approaches is that the fixed effect approach is much less sensitive to the inclusion of baseline covariates, whereas the IPTW sees relatively large changes in the point estimates across those specifications.

\section{Conclusion}\label{s:conclusion}

In this paper, we have show how it is possible to control for time-constant unmeasured confounding in marginal structural models by using a fixed effects approach to estimate the propensity score of the time-varying treatment. We derived the large-sample properties of this estimator under an asymptotic setup where the number of time periods and the number of units grows together. Simulations showed that the proposed method outperforms a naive approach that omits fixed effects and performs well overall especially when the magnitude of the heterogeneity is moderate. We applied this approach to estimating the time-varying effect of negative advertising on election outcomes in United States statewide elections and found that the fixed effect approach led to different conclusions about how negativity affects vote shares. An obvious place for future research would be to apply these methods to data where we have repeated measurements of the outcomes as well as the treatment. In those situations it may be possible to develop doubly-robust estimators under fixed effects assumptions.

\bibliographystyle{chicago}
\bibliography{psfe}

\clearpage
\appendix
\input{psfe_appendix.tex}

\end{document}

%% file: psfe_simulation.tex
In this section, we conduct simulation studies to evaluate the finite sample performance of the proposed approach.

\subsection{Setup}

We simulate a balanced panel of $n$ units with $T$ time points
where the number of units varies $n \in \{200, 500, 1000, 3000\}$.
We fix the ratio of the number of units to the number of time periods $ n/T = \rho \in \{5, 10, 50\}$.
This setup mimics the key asymptotic approach of our theoretical results, and the larger value of $\rho$ implies the small number of time points, $T  = n / \rho$.
The treatment sequence is generated as a function of individual unobserved effect $\alpha_{i}$, the past treatment $D_{i,t-1}$
and the time-varying covariates, $\*X_{it}$.
\begin{equation*}
D_{it} \sim \text{Bernoulli}(
  \text{expit}(\alpha_{i} + \varphi D_{i,t-1} + \bm{\beta}^{\top}\mathbf{X}_{it})
)
\end{equation*}
where $\text{expit}(x) = 1 / (1 + \exp(-x))$ is the inverse logistic function.
The individual heterogeneity is drawn from a uniform distribution with support on $[-a, a]$ for $a \in \{1, 2\}$.
The value of $a$ is chosen such that the variance of individual heterogeneity explains
$25\%$ ($a = 1$) or $50\%$ ($a = 2$) of the  variance of the linear predictor.
The time-varying covariates $\*X_{it}$ are generated exogenous to the treatment, drawn from the multivariate normal distribution,
$\*X_{it} \sim \mathcal{N}(-1/2\bm{1}, \Sigma)$
where $\Sigma_{jj} = 1$ and $\Sigma_{jj'} = 0.2$ for $j \neq j'$.
Finally, we set $\varphi = 0.3$ and $\bm{\beta} = (-0.5, -0.5)$
when the number of covariates is two or $\bm{\beta} = (-0.5, -0.5,  1.0, -0.5)$ when the number of covariates is four.

The outcome is generated by the linear model
with individual unobserved variable $\alpha_{i}$,
the final treatment $D_{iT}$,
the cumulative treatments $\sum^{T-3}_{t = T-1}D_{it}$
and the average of the time-varying covariates, $\overline{\*X}_{i} = \sum^{T}_{t=1}\*X_{it}/T$,
all of which are generated in the previous step.
\begin{equation*}
Y_{i} = \alpha_{i} +
        \tau_{F}D_{iT} +
        \tau_{C}\sum^{T-3}_{t = T-1}D_{it} +
        \bm{\gamma}^{\top}\overline{\mathbf{X}}_{i} +
        \epsilon_{i},\quad
        \epsilon_{i} \sim \mathcal{N}(0, 1)
\end{equation*}
where we set $\tau_{F} = 1$, $\tau_{C} = 0.3$, and $\bm{\gamma} = (1.0, 0.5)$ or
$\bm{\gamma} = (1.0, 0.5,  1.0,  1.0)$ depending on the number of covariates used in each simulation.

\begin{figure}[H]
  \centerline{
  \includegraphics[width=\textwidth]{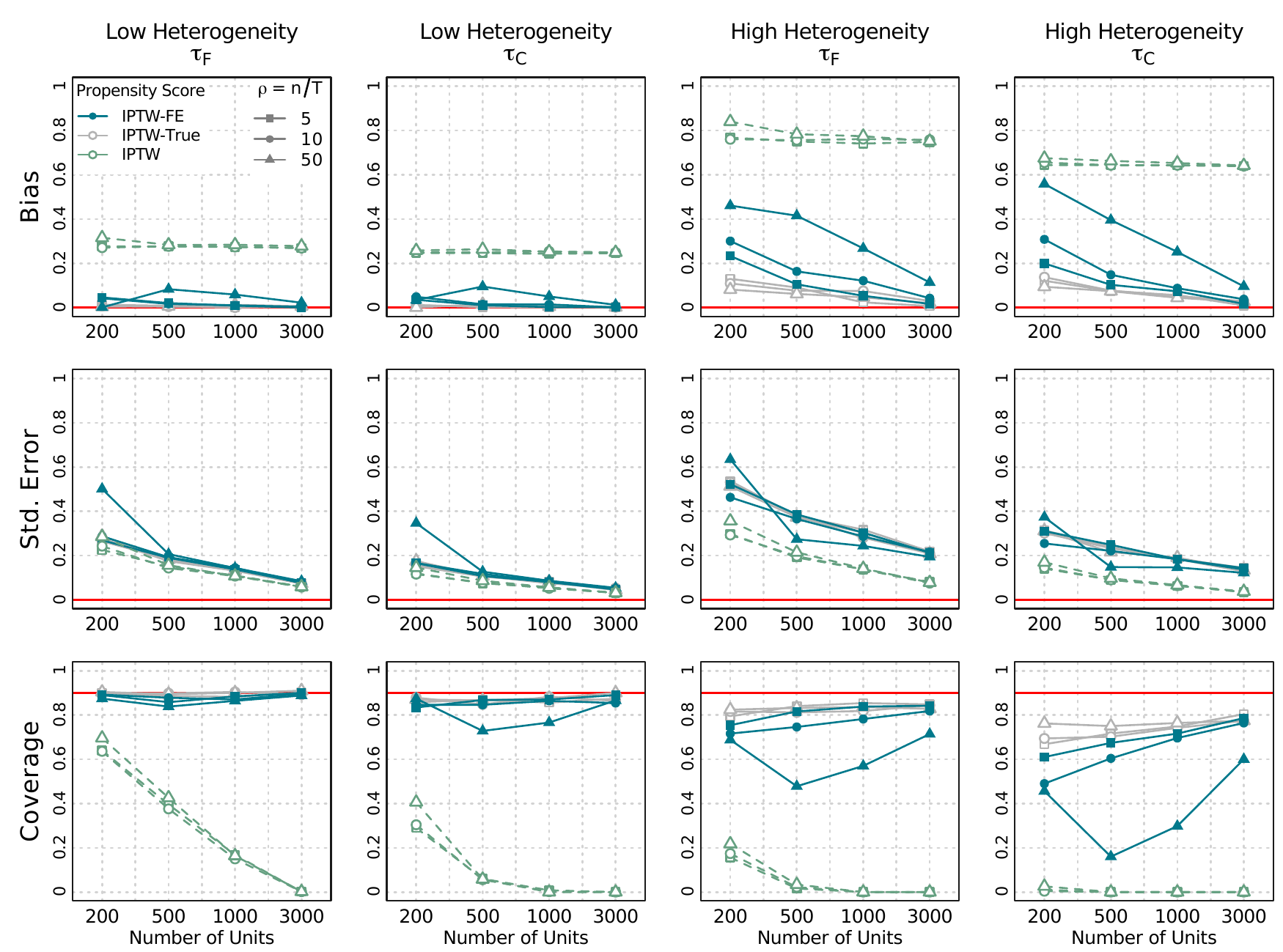}
  }
  \caption{
  Bias, standard error (Std. Error) and coverage probability of 90\% confidence intervals (Coverage) for the estimation of the final period effect $\tau_{F}$ and the cumulative effect $\tau_{C}$ under the ``low'' heterogeneity ($a = 1$) -- first two columns --
  and the ``high'' heterogeneity ($a = 2$) -- last two columns -- scenario.
  Solid lines in \textcolor{simblue}{blue} show the proposed estimator
  (\IPTWFE),
  solid lines in \textcolor{simgrey}{grey} show the esti\&ator based on the true propensity score (\IPTWT), and
  dashed lines in \textcolor{simgreen}{green} show the estimator based on the estimated propensity score without fixed effects (\IPTW).
  Shapes correspond to the $n$ to $T$ ratio $\rho$ such that
  squares represent $\rho = 5$ (the largest number of time periods),
  circles represent $\rho = 10$, and
  triangles represent $\rho  = 50$ (the smallest number of time periods)
  }
  \label{fig:msm-cov2}
\end{figure}

\subsection{Results}

We compared the performance of the proposed method in terms of estimating two causal quantities: the final period effect $\tau_{F}$
and the cumulative effect $\tau_{C}$.
We estimate two quantities together in the framework of weighted least square,
\begin{equation*}
(\widehat{\tau}_{F}, \widehat{\tau}_{C}) = \argmin_{\tau_{F}, \tau_{C}}
\sum^{n}_{i=1} \widehat{W}_{i}
\bigg\{
  Y_{i} - \alpha - \tau_{F}D_{iT} - \tau_{C}\sum^{T-3}_{t = T-1}D_{it}
\bigg\}^{2}
\end{equation*}
where $\widehat{W}_{i}$ is constructed as described in the previous section.
The variance of $\widehat{\tau}_{F}$ and $\widehat{\tau}_{C}$ is estimated using the standard sandwich formula with the HC2 option.

In addition to the fixed effect approach,
we consider two other strategies to obtain the weights $\widehat{W}_{i}$
as benchmarks to the proposed method.
First, we use the true propensity score to construct the weights.
Second, the estimated propensity score without the fixed effect
is used to construct weights.
We expect that the weights with known propensity scores
is least biased and the weights without the fixed effect is most biased.

Figure~\ref{fig:msm-cov2} shows the results
for the two-covariate case.
Bias (first row), standard errors (second row)
and coverages (third row) are computed based on  500 Monte Carlo simulations.
Additional simulation results are presented in Appendix~\ref{appx:additional-simulation}.
The first two columns correspond to the ``low'' heterogeneity case
where the support of the fixed effect is $[-1, 1]$,
whereas the last two columns correspond to the ``high'' heterogeneity
scenario where the support of $\alpha_{i}$ is set to $[-2, 2]$.
Solid lines in blue show the proposed estimator
(\IPTWFE),
solid lines in grey show the estimator based on the true propensity score (\IPTWT), and
dashed lines in green show the estimator based on the estimated propensity score without fixed effects (\IPTW).
Shapes correspond to the $n$ to $T$ ratio $\rho$ such that
squares represent $\rho = 5$ (the largest number of time periods),
circles represent $\rho = 10$, and
triangles represent $\rho  = 50$ (the smallest number of time periods).

We can see that under the low heterogeneity setting,
where the unobserved individual heterogeneity explains roughly 25\% of the variance of the treatment assignment,
the bias of the proposed estimator (\IPTWFE) is indistinguishable from the estimator that is based on the true propensity score (\IPTWT)
and the confidence interval estimates maintain the nominal coverage across different values of $n$ and $\rho$.
Under this scenario, even a case of $n = 200$ and $T = 4$,
the proposed method performs well.

When the variance of the individual heterogeneity is high ($a = 2$)
such that it explains roughly 75\% of the variance of the treatment assignments, the proposed estimator shows relatively larger bias compared with \IPTWT, while bias of the estimator without fixed effects (\IPTW) is substantially larger.
Under this setting, the coverage results are mainly driven by the bias, thus the figure shows that as $n$ increases, the coverage results also improves thanks to the reduction in bias.
We can also see that in general the estimator without fixed effects
shows smaller standard errors than \IPTWFE.
This implies that the proposed method (\IPTWFE) trade-off the efficiency with lower bias.
Finally, we highlight that small Monte Carlo bias is observed even for \IPTWT\ under this scenario. This is possibly due to the high variability of the weights, which are produce of inverse probabilities over four time periods with stabilization.

Overall, these results point to two key tensions in controlling for time-constant unmeasured heterogeneity through fixed effects in the propensity score models. First, high degrees of unmeasured heterogeneity in the propensity scores may lead to near violations of the positivity assumption that could lead to the kind of instability we see when $a = 2$. Second, larger magnitudes of heterogeneity may require more time periods to achieve good finite sample performance compared to when the heterogeneity is relatively small.

%% file: psfe_appendix.tex
\section{Proofs}

\input{proof_proposition.tex}

\input{proof_theorem.tex}

\clearpage
\section{Supporting Lemmas}
\label{appendix:lemmas}
\input{psfe_lemmas.tex}

\clearpage
\section{Additional Simulation Results}
\label{appx:additional-simulation}

\begin{figure}[htb]
\centerline{
  \includegraphics[width=\textwidth]{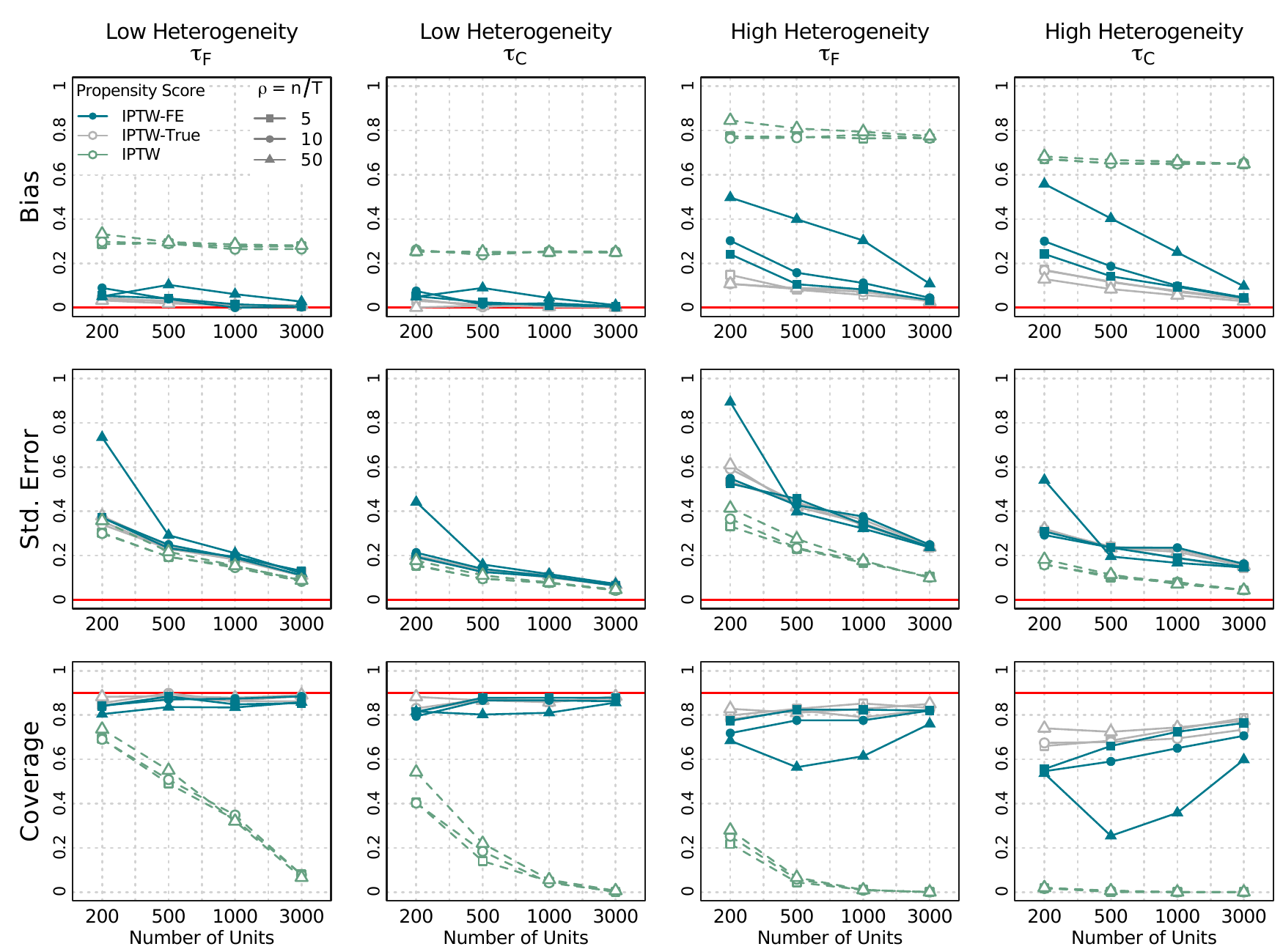}
}
\caption{
  Bias, standard error (Std. Error) and coverage probability of 90\% confidence intervals (Coverage) for the estimation of the final period effect $\tau_{F}$ and the cumulative effect $\tau_{C}$ under the ``low'' heterogeneity ($a = 1$) -- first two columns --
  and the ``high'' heterogeneity ($a = 2$) -- last two columns -- scenario.
  The number of time-varying covariates is four.
  Solid lines in \textcolor{simblue}{blue} show the proposed estimator
  (\IPTWFE),
  solid lines in \textcolor{simgrey}{grey} show the estimator based on the true propensity score (\IPTWT), and
  dashed lines in \textcolor{simgreen}{green} show the estimator based on the estimated propensity score without fixed effects (\IPTW).
  Shapes correspond to the $n$ to $T$ ratio $\rho$ such that
  squares represent $\rho = 5$ (the largest number of time periods),
  circles represent $\rho = 10$, and
  triangles represent $\rho  = 50$ (the smallest number of time periods)
}
\end{figure}

\begin{figure}[htb]
\centerline{
  \includegraphics[scale=0.75]{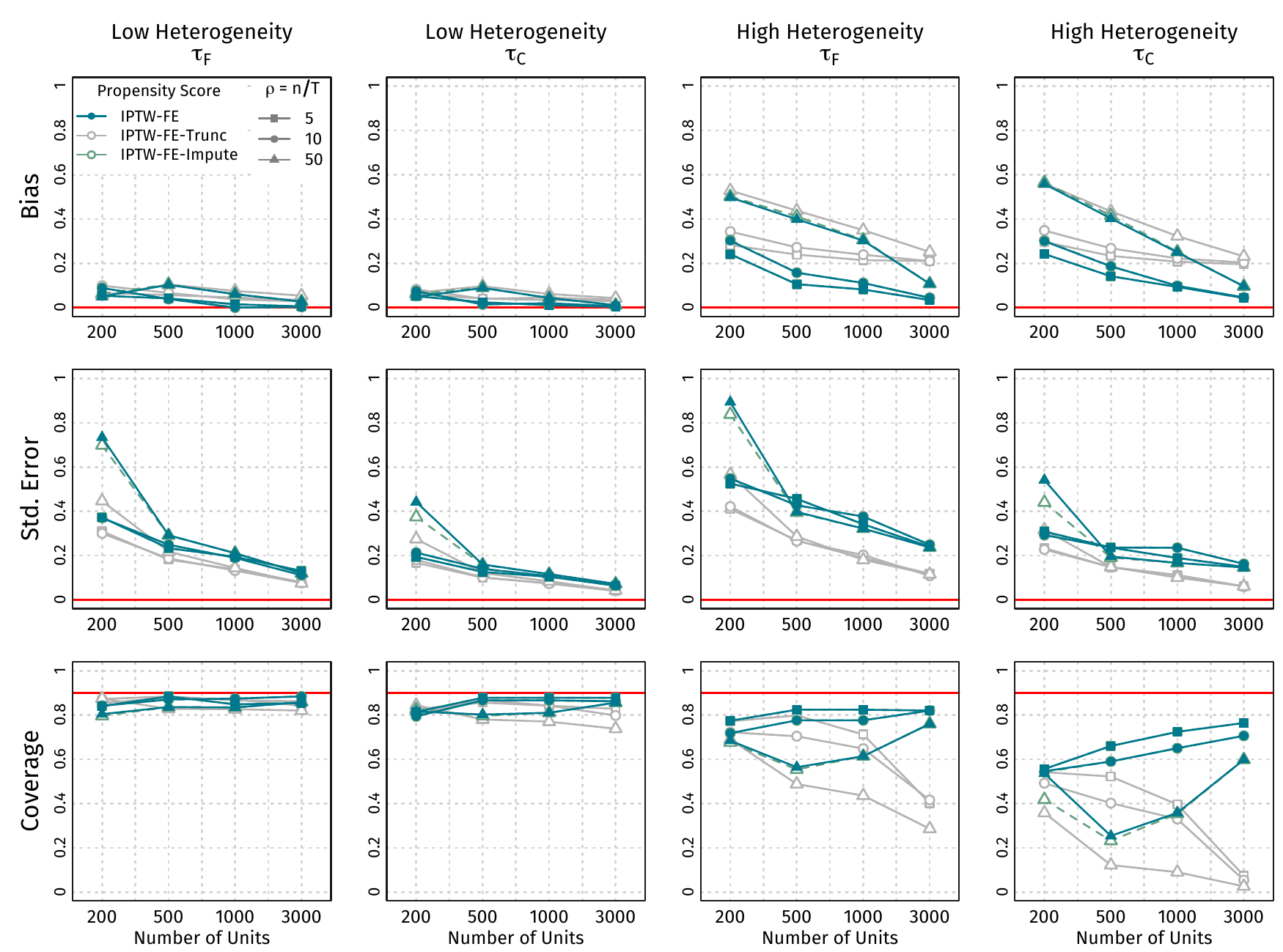}
}
\centerline{
  \includegraphics[scale=0.75]{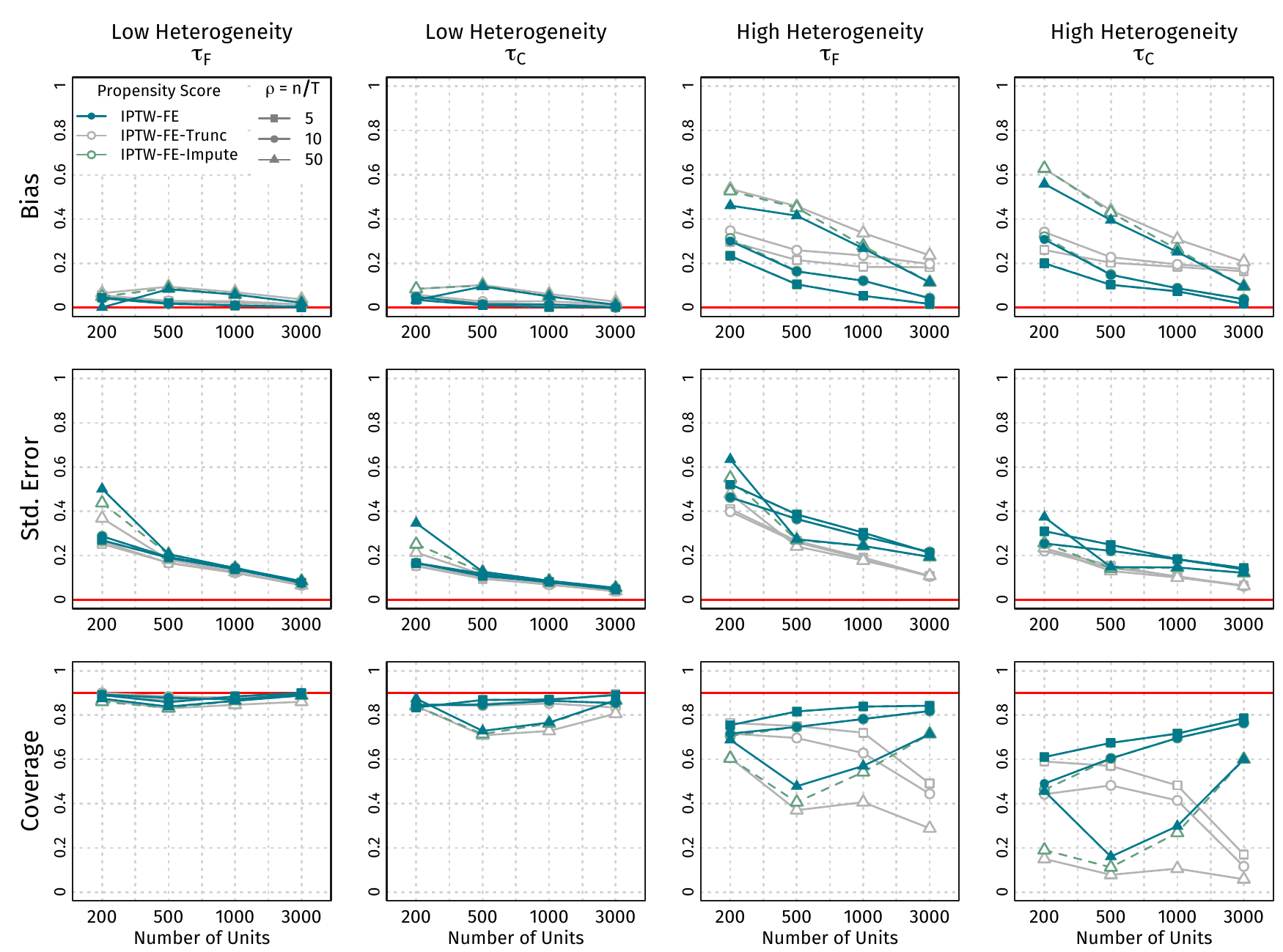}
}
\caption{Simulation results for imputing the non-identified fixed effect estimates.}
\end{figure}

%% file: proof_proposition.tex

\subsection{Unit-specific Randomized Experiments}
\label{appendix:proof-prop-simple}

\begin{proof}[Proposition~\ref{prop:simple}]

Under standard expansion of M-estimators, we have
\begin{align*}
0
&= \frac{1}{N}\sum^{N}_{i=1}U_{i}(\widehat{\pi}_{i}, \widehat{\tau}) \\
&= \frac{1}{N}\sum^{N}_{i=1}U_{i}(\widehat{\pi}_{i}, \tau)
    + \frac{1}{N}\sum^{N}_{i=1} \frac{\partial}{\partial \tau}
        U_{i}\bigg|_{\tau = \tau_{0}, \pi = \widehat{\pi}_{i}}
        (\widehat{\tau} - \tau) \\
&= \frac{1}{N}\sum^{N}_{i=1}U_{i}(\widehat{\pi}_{i}, \tau)
    - \frac{1}{N}\sum^{N}_{i=1} \frac{D_{iT}}{\widehat{\pi}_{i}}
        (\widehat{\tau} - \tau)
\end{align*}

Then, we have
\begin{align*}
\sqrt{N}(\widehat{\tau} - \tau)
&= \Bigg(
  \frac{1}{N}\sum^{N}_{i=1} \frac{D_{iT}}{\widehat{\pi}_{i}}
\Bigg)^{-1}
\frac{1}{\sqrt{N}}\sum^{N}_{i=1}U_{i}(\widehat{\pi}_{i}, \tau)
\end{align*}

Under the standard Taylor expansion of $\widehat{\pi}_i$ around the true value $\pi_i$ we have the following:
\[
\frac{1}{\sqrt{N}}\sum^{N}_{i=1}U_{i}(\widehat{\pi}_{i}, \tau) = \underbrace{\frac{1}{\sqrt{N}} \sum_{i=1}^N U_i}_{\textrm{(I)}} - \underbrace{\frac{1}{\sqrt{N}} \sum_{i=1}^N \frac{(\widehat{\pi}_i - \pi_i)}{\pi_i} U_i}_{\textrm{(II)}} + \underbrace{\frac{1}{\sqrt{N}} \sum_{i=1}^N \frac{(\widehat{\pi}_{i}- \pi_i)^2}{\overline{\pi}^3}D_{iT}(Y_i(1) - \tau)}_{\textrm{(III)}}
\]
where $\overline{\pi}_i$ is a mean value between $\widehat{\pi}_i$ and $\pi_i$. The first term, (I), in this expansion is just the influence function for the infeasible estimator and so converges to $\mathcal{N}(0, \E[U^{2}_{i}])$.

For (II), we first rewrite this as:
\[
  \begin{aligned}
    \text{(II)} &= \sqrt{N}\left( \frac{1}{NT} \sum_{i=1}^N\sum_{t=1}^T \frac{D_{it}-\pi_i}{\pi_i}U_i \right)
    \end{aligned}
\]
The potential outcomes only depend on the last  period, which implies that $Y_{i}(d) \indep D_{it} \mid \alpha_i$ for all $t$, so we have
\begin{align*}
  \E\left\{\frac{D_{it}-\pi_i}{\pi_i}U_i  \right\}
  & = \E\left\{ \frac{D_{iT}(D_{it}-\pi_i)}{\pi^2_i}(Y_i(1)-\tau) \right\} \\
  & = \E\left\{ \frac{1}{\pi_i} \E[(D_{it}-\pi_i)(Y_i(1)-\tau) \mid \alpha_i] \right\} \\
  &= \E\left\{  \frac{1}{\pi_i}\E[(D_{it}-\pi_i)\mid \alpha_i]\E[(Y_i(1)-\tau) \mid \alpha_i] \right\} = 0,
\end{align*}
where the second equality follows from sequential ignorability (i.e., applying the g-computational formula in this case) and the fourth equality uses independence of $\E[D_{it} - \pi_i\mid \alpha_i] = 0$.

Thus, (II) is a mean-zero random variable so we have (II) $=O_P(\V(\text{(II)})^{1/2})$. By bounded propensity scores, bounded outcome moments and standard arguments, we can show that $\V(\text{(II)}) = O(1/T)$ so that (II) is $O_p(1/\sqrt{T})$. 

For (III), we first note that $\max_i|\widehat{\pi}_i - \pi_i|^2 = O_p(1/T)$ and so we have
\[
  \begin{aligned}
    |(\text{III})| &\leq  \frac{1}{\sqrt{N}} \sum_{i=1}^N \left|   \frac{(\widehat{\pi}_{i}- \pi_i)^2}{\overline{\pi}^3}D_{iT}(Y_i(1) - \tau)\right| \\
    & \leq \sqrt{N}\max_i|\widehat{\pi}_i - \pi_i|^2 \frac{1}{N} \sum_{i=1}^N \left|   \overline{\pi}^{-3}D_{iT}(Y_i(1) - \tau)\right| \\
  \end{aligned}
\]
For the last term here, we note that
\[
  \begin{aligned}
    &\E\left[ \sup_{\pi_i \in \mathcal{B}(\epsilon)} \left( \frac{1}{N}\sum_{i=1}^N\left|   \overline{\pi}^{-3}D_{iT}(Y_i(1) - \tau)\right|  \right)^2 \right]    \\
    &\leq \E\left[ \left( \frac{1}{N}\sum_{i=1}^N\left|D_{iT}(Y_i(1) - \tau)\right|  \right)^2 \right]  && \text{(bounded PS)} \\
    &\leq \E\left[\left( \frac{1}{N}\sum_{i=1}^N\left|(Y_i(1) - \tau)\right|  \right)^2 \right] < \infty  && \text{(bounded moments)}
  \end{aligned}
\]
Note we are abusing notation slightly here by taking the supremum over all the propensity scores across all units.
By Markov inequality, this implies that
\[
  \sup_{\pi_i \in \mathcal{B}(\epsilon)} \frac{1}{N} \sum_{i=1}^N \left|\pi^{-3}D_{iT}(Y_i(1) - \tau)\right| = O_p(1),
\]
From here, we have that
\[
  \frac{1}{N} \sum_{i=1}^N \left|\overline{\pi}^{-3}D_{iT}(Y_i(1) - \tau)\right| = O_p(1),
\]
where $\overline{\pi}_i \in \mathcal{B}(\epsilon)$ with probability approaching 1 since it is between $\widehat{\pi}_i$ and $\pi_i$. Thus, we have
\[
|(\text{III})| \leq \sqrt{N}\max_i|\widehat{\pi}_i - \pi_i|^2 \frac{1}{N} \sum_{i=1}^N \left|   \overline{\pi}^{-3}D_{iT}(Y_i(1) - \tau)\right| = \sqrt{N}O_p(1/T)O_P(1) = O_p(1/\sqrt{T})
\]
All together, this implies that $\sqrt{N}(\widehat{\tau} - \tau_1) \xrightarrow{d} N(0, V)$.

\end{proof}

%% file: proof_theorem.tex

\subsection{Marginal Structural Models}\label{appendix:proof-thm-iptwfe}

\begin{proof}[Proof of Theorem~\ref{thm:msm-iptw-fe}]

Suppose now we are interested in an MSM $g(\underline{d}_k;\gamma)\E[Y_{iT}(\underline{d}_k)]$, where $\underline{d}_k = (d_{T-k},\ldots, d_T)$ and $k$ is fixed. The parameter vector $\gamma$ is of length $J$. Define the probability of a particular treatment history as a function of the propensity score parameters as
\[
W_{i}(\underline{d}_k; \beta, \alpha_i) = \prod_{j=0}^k \pi_{i,T-j}(\beta, \alpha_i)^{d_{T-j}}(1 - \pi_{i,T-j}(\beta, \alpha_i))^{1 - d_{T-j}}.
\]
Generally, we can define an MSM as the solution to the following:
\[
0 = \E\left\{ \frac{h(\underline{D}_{ik})(Y_i - g(\underline{D}_{ik}; \gamma))}{W_i(\underline{D}_{ik}; \beta, \alpha_i)} \right\}
\]
where $h(\cdot)$ is a function with $J$-length output, chosen by the researcher.  For example, if $Y_i$ is continuous and $g$ is linear and additive, it is common to use $h(\underline{D}_{ik}) = \underline{D}'_{ik}$. If we knew the propensity scores, we could derive an estimator of $\gamma$ based on the sample moment condition:
\[
  0 = \frac{1}{N} \sum_{i=1}^N \frac{h(\underline{D}_{ik})(Y_i - g(\underline{D}_{ik};\widehat{\gamma}))}{W_i(\underline{D}_{ik};\beta, \alpha_i)} = \frac{1}{N} \sum_{i=1}^N U_i(\widehat{\gamma}, \beta, \alpha_i)
\]
Because the propensity score is never known in observational studies, we define our estimator based on the estimated propensity scores:
\[
0 = \frac{1}{N} \sum_{i=1}^N \frac{h(\underline{D}_{ik})(Y_i - g(\underline{D}_{ik};\widehat{\gamma}))}{W_i(\underline{D}_{ik};\widehat{\beta}, \widehat{\alpha}_i)} = \frac{1}{N} \sum_{i=1}^N U_i(\widehat{\gamma}, \widehat{\beta}, \widehat{\alpha}_i)
\]

\paragraph{Consistency}

Note that within $(\beta, \alpha_i) \in \mathcal{B}_0(\epsilon)$,  $|\partial_{\beta} \ell_{it}(\beta, \alpha_i)| < M(Z_{it})$ implies that $|\partial_{\beta_k} \pi_{it}(\beta, \alpha_i)| < CM(Z_{it})$ for some constant $C$ since,
\[
|\partial_{\beta_k} \ell_{it}(\beta, \alpha_i)| = \abs*{\left(  \frac{D_{it} - \pi_{it}(\beta, \alpha_i)}{\pi_{it}(\beta, \alpha_i)(1 - \pi_{it}(\beta, \alpha_i))}\right) \partial_{\beta_k} \pi_{it}(\beta, \alpha_i)},
\]
and the propensity scores are uniformly bounded from below over $\mathcal{B}_0(\epsilon)$. The same applies to $|\partial_{\alpha} \pi_{it}(\beta, \alpha)|$. With these results, we can also bound the partial derivatives of the weights:
\[
  \begin{aligned}
    |\partial_{\beta_k} W_i(\beta, \alpha_i)| &= \abs*{\sum_{t=T-k}^T \partial_{\beta_k} \pi_{it}(\beta, \alpha_i) \sum_{s \neq t} \pi_{is}(\beta, \alpha_i)} \leq k(k - 1) |\partial_{\beta_k} \pi_{it}(\beta, \alpha_i)| \leq Ck(k-1)M(Z_{it}),
  \end{aligned}
\]
where $k$ is fixed as $N,T \to \infty$. Again, a similar expression holds for $|\partial_{\alpha}W_i(\beta, \alpha_i)|$. Thus, by the mean value theorem, we have
\[
  \begin{aligned}
    |W_i(\widehat{\beta}, \widehat{\alpha}_i) - W_i(\beta_0, \alpha_{i0})| &\leq \norm{\partial_{\beta} W_i(\beta, \alpha_i)}\;\norm{\widehat{\beta}-\beta_0} + | \partial_{\alpha} W_i(\beta, \alpha_i)| \; |\widehat{\alpha}_i - \alpha_{i0}| \\
    &\leq C_{\beta}M(Z_{it})\norm{\widehat{\beta}-\beta_0} + C_{\alpha}M(Z_{it})|\widehat{\alpha}_i - \alpha_{i0}|,
  \end{aligned}
\]
for some constants $C_{\beta}$ and $C_{\alpha}$.

Using this result we can uniformly bound the convergence of the estimating equation in terms of the parameters of the weighting model.
  \begin{align*}
    \sup_{\gamma \in \Gamma} &|N^{-1} \sum_{i=1}^N U_i(\gamma, \widehat{\beta}, \widehat{\alpha}_i) - U_i(\gamma, \beta_0, \alpha_{i0})|\\
    & \leq  N^{-1} \sum_{i=1}^N \sup_{\gamma \in \Gamma} | U_i(\gamma, \widehat{\beta}, \widehat{\alpha}_i) - U_i(\gamma, \beta_0, \alpha_{i0})| \\
    & \leq N^{-1} \sum_{i=1}^N \sup_{\gamma \in \Gamma}|h(\underline{D}_{ik})(Y_i - g(\underline{D}_{ik}, \gamma))(W_i(\widehat{\beta}, \widehat{\alpha}_i)^{-1} - W_i(\beta_{0}, \alpha_{i0}))| \\
    & \leq N^{-1} \sum_{i=1}^N \frac{|W_i(\widehat{\beta}, \widehat{\alpha}_i) - W_i(\beta_{0}, \alpha_{i0})|}{W_i(\widehat{\beta}, \widehat{\alpha}_i)W_i(\beta_{0}, \alpha_{i0})}\sup_{\gamma \in \Gamma}|h(\underline{D}_{ik})(Y_i - g(\underline{D}_{ik}, \gamma))| \\
    & < N^{-1} \sum_{i=1}^N |W_i(\widehat{\beta}, \widehat{\alpha}_i) - W_i(\beta_{0}, \alpha_{i0})|\sup_{\gamma \in \Gamma}|h(\underline{D}_{ik})(Y_i - g(\underline{D}_{ik}, \gamma))| \\
    & \leq C_{\beta}\norm{\widehat{\beta} - \beta_0} N^{-1} \sum_{i=1}^N |M(Z_{it})|
      \sup_{\gamma \in \Gamma}|h(\underline{D}_{ik})(Y_i - g(\underline{D}_{ik}, \gamma))| \\
    &\qquad +  C_{\alpha}\max_i\abs*{\widehat{\alpha}_i - \alpha_{i0}} N^{-1}
      \sum_{i=1}^N |M(Z_{it})|\sup_{\gamma \in \Gamma}|h(\underline{D}_{ik})(Y_i - g(\underline{D}_{ik}, \gamma))|
  \end{align*}

The fourth inequality holds because under Assumption~\ref{a:bounded-propensity-scores}, we have that $W_i(\underline{d}_k;\beta,\alpha) \in (\epsilon, 1 - \epsilon)$ where $\epsilon > 0$ near $(\beta_0, \alpha_{i0})$. By the bounded moments of $M(Z_{it})$ and $Y_i$, we have that $N^{-1} \sum_{i=1}^N |M(Z_{it})|\sup_{\gamma \in \Gamma}|h(\underline{D}_{ik})(Y_i - g(\underline{D}_{ik}, \gamma))| = O_p(1)$. Combined with the consistency of $\widehat{\beta}$ and $\widehat{\alpha}_i$ from Lemma~\ref{lem:alpha-expansion}, we have $\sup_{\gamma \in \Gamma} |N^{-1} \sum_{i=1}^N U_i(\gamma, \widehat{\beta}, \widehat{\alpha}_i) - U_i(\gamma, \beta_0, \alpha_{i0})| = o_p(1)$.
Thus, we have
\[
  \begin{aligned}
    |N^{-1}\sum_{i=1}^N U_i(\widehat{\gamma}, \widehat{\beta}, \widehat{\alpha}_i)|  &\leq    |N^{-1}\sum_{i=1}^N U_i(\widehat{\gamma}, \beta_0, \alpha_{i0})| +  |N^{-1}\sum_{i=1}^N U_i(\widehat{\gamma}, \widehat{\beta}, \widehat{\alpha}_i) - U_i(\widehat{\gamma}, \beta_0, \alpha_{i0})| \\
    &\leq |N^{-1}\sum_{i=1}^N U_i(\widehat{\gamma}, \beta_0, \alpha_{i0})| + \sup_{\gamma \in \Gamma}|N^{-1} \sum_{i=1}^N U_i(\gamma, \widehat{\beta}, \widehat{\alpha}_i) - U_i(\gamma, \beta_0, \alpha_{i0})| \\
    &= |N^{-1}\sum_{i=1}^N U_i(\widehat{\gamma}, \beta_0, \alpha_{i0})| + o_p(1).
  \end{aligned}
\]
This establishes that $\widehat{\gamma} \inprob \gamma_0$, by standard results of estimating equations.

\paragraph{Asymptotic expansion}

Let $G_{i}(\gamma, \beta, \alpha) = \partial U_i(\gamma, \beta, \alpha) / \partial \gamma$, then we have the following expansion:
\[
  0 = \frac{1}{\sqrt{N}} \sum_{i=1}^N U_i(\widehat{\gamma}, \widehat{\beta}, \widehat{\alpha}_i) = \frac{1}{\sqrt{N}} \sum_{i=1}^N U_i(\gamma_0, \widehat{\beta}, \widehat{\alpha}_i) + \sqrt{N}(\widehat{\gamma} - \gamma_0)\left( \frac{1}{N} \sum^{N}_{i=1}G_i(\overline{\gamma}, \widehat{\beta}, \widehat{\alpha}_i) \right)
\]
where $\overline{\gamma}$ is a value between $\widehat{\gamma}$ and $\gamma_0$. This implies the following influence-function representation for the estimator:
  \begin{align}
  \sqrt{N}(\widehat{\gamma} - \gamma_0) &= \left( \frac{1}{N} \sum_{i=1}^N G_i(\overline{\gamma}, \widehat{\beta}, \widehat{\alpha}_i) \right)^{-1} \frac{1}{\sqrt{N}} \sum_{i=1}^N U_i(\gamma_0, \widehat{\beta}, \widehat{\alpha}_i)  \\
  &= G^{-1} \frac{1}{\sqrt{N}} \sum_{i=1}^N U_i(\gamma_0, \widehat{\beta}, \widehat{\alpha}_i) + o_p(1)\label{eq:gamma-expand},
\end{align}
where $G = \E[G_i]$ and noting that function without arguments are evaluated at the true values of the parameters, $G_i = G_i(\gamma_0, \beta_0, \alpha_{i0})$. The second equality here follows from Lemma 2.4 of~\citet{NewMcF94} after noting  that $\overline{\gamma}$ is between $\widehat{\gamma}$ and $\gamma_0$, that $\widehat{\gamma}$, $\widehat{\beta}$, and $\widehat{\alpha}_i$ are all consistent and from Assumption~\ref{a:msm-assumptions}.

We can expand the $r$th element of $U_{i}$, $U_{ir}$ in Equation~\eqref{eq:gamma-expand} as
\begin{align*}
\frac{1}{\sqrt{N}}\sum_{i=1}^N U_{ir}(\gamma_{0}, \widehat{\beta}, \widehat{\alpha}_{i})
&=
\frac{1}{\sqrt{N}}\sum_{i=1}^N U_{ir}
+ \underbrace{
  \sqrt{N}(\widehat{\beta} - \beta_{0})^{\top} \left(\frac{1}{N}\sum_{i=1}^N\frac{\partial}{\partial \beta}U_{ir}
    \right)}_{\text{(I)}}
+ \underbrace{
  \frac{1}{\sqrt{N}} \sum_{i=1}^N (\widehat{\alpha}_{i} - \alpha_{i0}) \frac{\partial}{\partial \alpha}U_{ir}
  }_{\text{(II)}} \\
&\quad
+ \underbrace{
  \sqrt{N}(\widehat{\beta} - \beta_{0})^{\top}
    \left( \frac{1}{N} \sum_{i=1}^N \frac{\partial^{2}}{\partial \beta \partial \beta}U_{ir}(\gamma_{0}, \overline{\beta}, \overline{\alpha}_{i}) \right)
  (\widehat{\beta} - \beta_{0})
  }_{\text{(III)}} \\
&\quad + \underbrace{
  \frac{1}{\sqrt{N}} \sum_{i=1}^N (\widehat{\alpha}_{i} - \alpha_{i})^{2}
  \frac{\partial^{2}}{\partial \alpha \partial \alpha}U_{ir}(\gamma_{0}, \overline{\beta}, \overline{\alpha}_{i})
  }_{\text{(IV)}}\\
&\quad
+ \underbrace{
  \sqrt{N}(\widehat{\beta} - \beta)^{\top}
  \left( \frac{1}{N} \sum_{i=1}^N \frac{\partial^{2}}{ \partial \beta \partial \alpha}U_{ir}(\gamma_{0}, \overline{\beta}, \overline{\alpha}_{i})
  (\widehat{\alpha}_{i} - \alpha_{i0}) \right)
}_{\text{(IV)}}
\end{align*}
where $\overline{\alpha}_{i}$ is the value between $\widehat{\alpha}_{i}$ and $\alpha_{i0}$ ($\overline{\beta}$ is defined similarly).

\paragraph{First order terms}
We first show that Term (I) and (II) in the above expression are both $o_{p}(1)$.

We denote derivatives of $U_{ir}$ with subscripts so that,
\begin{align*}
U_{ir,\alpha}(\gamma, \beta, \alpha) = \frac{\partial}{\partial \alpha}U_{ir}(\gamma, \beta, \alpha), \quad
\text{and}\quad
U_{ir,\beta}(\gamma, \beta, \alpha) = \frac{\partial}{\partial \beta}U_{ir}(\gamma, \beta, \alpha)
\end{align*}
and simply write $U_{ir,\alpha}$ when evaluated at the true values of the parameters. Letting $h(\underline{D}_{ik})_{[r]}$ be the $r$th entry of that vector, the expression of $U_{ir,\alpha}$ is given by
\begin{align*}
U_{ir,\alpha}
&= \frac{h(\underline{D}_{ik})_{[r]} (Y_{i} - g(\underline{D}_{ik}; \gamma_{0}))}
        {W^{2}_{i}(\underline{D}_{ik}; \beta, \alpha_{i})}\\
&\quad \times
  \sum^{T}_{t=T-k}\bigg\{
  (2D_{it}-1)\frac{\partial}{\partial \alpha_{i}} \pi_{it}(\beta_{0}, \alpha_{i0})
  \prod_{\substack{s = T-k \\  s \neq t}}^T\pi_{is}(\beta_{0}, \alpha_{i0})^{D_{is}}
  \big[1 - \pi_{is}(\beta_{0}, \alpha_{i0}) \big]^{1 - D_{is}}
  \bigg\} \\
&= U_{ir}
\sum^{T}_{t=T-k}
\frac{(2D_{it} - 1)\pi_{it}V_{it}}
      {\pi^{D_{it}}_{it}(1 - \pi_{it})^{1-D_{it}}}\\
&= U_{ir} \sum^{T}_{t=T-k} \bigg[
  (2D_{it}-1) V_{it} \bigg\{ \frac{\pi_{it}}{1 - \pi_{it}} \bigg\}^{1-D_{it}}
\bigg]\\
&\equiv U_{ir}\overline{V}_{i}
\end{align*}
Here, we have used
\begin{align*}
\frac{\partial}{\partial \alpha_{i}} \pi_{it}(\beta_{0}, \alpha_{i0})
&= \pi_{it} \frac{\partial}{\partial \alpha_{i}} \log\ell_{it} = \pi_{it}V_{it}.
\end{align*}
Similarly, the expression of $U_{ir,\beta}$ can be derived as
\begin{align*}
U_{ir,\beta}
&=
U_{ir} \sum^{T}_{t=T-k} \bigg[
  (2D_{it}-1) S_{it} \bigg\{ \frac{\pi_{it}}{1 - \pi_{it}} \bigg\}^{1-D_{it}}
\bigg]\\
&\equiv U_{ir}\overline{S}_{i}
\end{align*}
where $S_{it}$ is the score function $\partial\log \ell_{it}/\partial \beta$.

To control (I), we use the following results:
\[
|\text{(I)}| \leq \sqrt{N}\Vert\widehat{\beta} - \beta_{0}\Vert \Bigg\Vert \frac{1}{N}\sum_{i=1}^N U_{ir}\overline{S}_i\Bigg\Vert
\]
Because $\sqrt{NT}\Vert \widehat{\beta} - \beta_{0} \Vert = O_p(1)$,  we have $\sqrt{N}\Vert\widehat{\beta} - \beta_{0}\Vert = O_p(1/\sqrt{T}) = o_p(1)$. Let $\overline{S}_{iq}$ be the $q$th entry in the $\overline{S}_i$ vector. Note that $\overline{S}_{iq}$ has bounded fourth moments by Lemma~\ref{lem:msm-second-order} since it a function of the $q$th score vector. Thus, for the second term bounding (I), we have:
\[
\E\left[ \left(\frac{1}{N} \sum_{i=1}^NU_{ir}\overline{S}_{iq}  \right)^2 \right]  \leq \frac{1}{N} \sum_{i=1}^N \E[(U_{ir}\overline{S}_{iq})^2] \leq \frac{1}{N} \sum_{i=1}^N \left( \E[U_{ir}^4] \right)^{1/2}\left( \E[\overline{S}_{iq}^4] \right)^{1/2} = O(1),
\]
where the first inequality holds because for any i.i.d. set of random variables $X_1,\ldots,X_n$, we have $\E[(n^{-1}\sum_i X_i)^2] \leq n^{-1}\sum_i \E[X_i^2]$. The second inequality holds by Cauchy-Schwarz, and the last equality holds by Assumption~\ref{a:bounded-derivatives} and~\ref{a:bounded-moments}. Because the same holds for all entries in $\overline{S}_i$, we have $\Vert N^{-1} \sum_{i=1}^N U_{ir}\overline{S}_i\Vert = O_p(1)$ by the Markov inequality and so (I) is $o_{p}(1)O_{p}(1) = o_{p}(1)$.

By Lemma~\ref{lem:alpha-expansion}, Term (II) can be written as
\begin{align*}
\text{(II)}
  &= \frac{1}{\sqrt{N}}\sum^{N}_{i=1}U_{ir,\alpha}(\widehat{\alpha}_{i} - \alpha_{i0}) = \frac{\sqrt{N}}{NT}\sum^{N}_{i=1}\sum^{T}_{t=1} U_{ir,\alpha}\psi_{it} + o_{p}(1)
  \\
  & = \underbrace{\frac{\sqrt{N}}{NT}\sum^{N}_{i=1}\sum^{T}_{t=1} \E[U_{ir,\alpha}\psi_{it} \mid \alpha_i]}_{\text{(II.a)}} + \underbrace{\frac{\sqrt{N}}{NT}\sum^{N}_{i=1}\sum^{T}_{t=1} \left( U_{ir,\alpha}\psi_{it} - \E[U_{ir,\alpha}\psi_{it} \mid \alpha_i] \right)}_{\text{(II.b)}}  + o_{p}(1),
\end{align*}
where the $o_p(1)$ term in the first line is due to $N^{-1}\sum_{i=1}^N U_{ir,\alpha}$ being $O_p(1)$ and the remained of the $\widehat{\alpha}_i$ expansion from Lemma~\ref{lem:alpha-expansion} being $\max_i|R_i| = o_p(T^{-1/2})$.

Note that $\psi_{it} = \E_T\{\E_{\alpha}[V_{it\alpha}]\}^{-1}V_{it}$.  In an abuse of notation, we define $\overline{V}_i(\underline{d}_k, D_{it})$ to be $\overline{V}_i$ with all covariates and the outcome replaced with their potential outcomes setting $\underline{d}_k$ and we leave $D_{it}$ as an argument to emphasize that this function depends on $D_{it}$. Furthermore, let $U_{ir,\alpha}(\underline{d}_k, D_{it}) = \overline{V}_i(\underline{d}_k, D_{it})h(d_k)_{[r]}(Y_i(\underline{d}_k, D_{it}) - g(\underline{d}_k;\gamma_0))$.  Then, applying the g-computational formula, we have for all $t < T - k$,
\[
\E[U_{ir,\alpha}\psi_{it} \mid \alpha_i] = \sum_{\underline{d}_k} \E_T\{\E_{\alpha}[V_{it\alpha}]\}^{-1} \E\left[V_{it}U_{ir,\alpha}(\underline{d}_k, D_{it}) \mid  \alpha_{i}  \right]
\]
Note that because $V_{it}$ is a score, we can use iterated expectations to show that $\E[V_{it}\mid \alpha_i]= 0$.
Thus, the inner expectation in the above expression is the covariance between $V_{it}$ and $U_{ir,\alpha}(\underline{d}_k, D_{it})$ conditional on $\alpha_i$. Using Lemma~\ref{lem:mixing-covariance}, we have
\[
  \begin{aligned}
    \abs*{
      \E\left[V_{it}U_{ir,\alpha}(\underline{d}_k, D_{it}) \mid  \alpha_{i}  \right]}
    & = \abs*{\text{Cov}(V_{it}, U_{ir,\alpha}(\underline{d}_k, D_{it}) \mid \alpha_i)} \\
    &  \leq 8 a(T-k-t)^{[1-1/(8+\nu)-1/(2+\nu)]} \\
    &\qquad \times \left[ \E[|V_{it}|^{8+\nu}] \right]^{1/(8+\nu)}\left[ \E[|U_{ir,\alpha}(\underline{d}_k, D_{it})|^{2+\nu}] \right]^{1/(2+\nu)} \\
    &  \leq C(T- k - t)^{-\mu[1-1/(8+\nu)-1/(2+\nu)]} \\
    &  \leq C(T - k - t)^{-4}
  \end{aligned}
\]
Thus, we have
\begin{align*}
  \abs*{\sum_{t=1}^T \E[U_{ir,\alpha}\psi_{it} \mid \alpha_i] } &\leq \sum_{t=1}^T \sum_{\underline{d}_k} \abs*{ \E_T\{\E_{\alpha}[V_{it\alpha}]\}^{-1} \E\left[V_{it}U_{ir,\alpha}(\underline{d}_k, D_{it}) \mid  \alpha_{i}  \right]}\\
  &\leq 2^k \sum_{t=1}^T\E\left\{\abs*{\E\left[V_{it}U_{ir,\alpha}(\underline{d}_k, D_{it}) \mid  \alpha_{i}  \right]}  \right\} \\
  & \leq C2^k \sum_{t=1}^T (T-k-t)^{-4} \\
  & \leq C2^k \sum_{m=1}^{\infty} m^{-4} = \frac{C2^k\pi^4}{90}
\end{align*}
Thus, we can establish that (II.a) is $O_p(1/\sqrt{T})$. By Lemma~\ref{lem:mixing-bounds} and the bounded moment conditions for the outcome and the partial derivatives, we have
\[
\frac{1}{NT} \sum_{i=1}^N\sum^{T}_{t=1} \left( U_{ir,\alpha}\psi_{it} - \E[U_{ir,\alpha}\psi_{it} \mid \alpha_i] \right) = O_p(1/\sqrt{NT}).
\]
This implies that (II.b) is $\sqrt{N}O_p(1/\sqrt{NT}) = O_p(T^{-1/2}) = o_p(1)$. Then, it follows that
\begin{equation*}
\text{(I)} + \text{(II)} = o_{p}(1) + o_{p}(1) = o_{p}(1).
\end{equation*}

\paragraph{Second order terms}

We will show that Term (III), (IV) and (V) are all $o_{p}(1)$.

Define the following second derivatives of the MSM estimating equation:
\[
U_{ir\alpha\alpha}(\gamma, \beta, \alpha) = \frac{\partial}{\partial \alpha} U_{ir\alpha}(\gamma, \beta, \alpha),
\quad
U_{ir\beta\beta}(\gamma, \beta, \alpha) = \frac{\partial}{\partial \beta}U_{ir\beta}(\gamma, \beta, \alpha),
\]
\[
U_{ir\beta\alpha}(\gamma, \beta, \alpha) = \frac{\partial}{\partial \alpha}U_{ir\beta}(\gamma, \beta, \alpha).
\]

Lemma~\ref{lem:alpha-expansion} implies estimation error in the fixed effects are uniformly bounded so (IV) is bounded as
\begin{equation}
  \label{eq:iv_bound}
|\text{(IV)}|
\leq
\sqrt{N}\max_i|\widehat{\alpha}_{i} - \alpha_{i}|^{2}
\bigg|\frac{1}{N}\sum^{N}_{i=1}U_{ir\alpha\alpha}(\gamma_{0}, \overline{\beta}, \overline{\alpha}_{i})\bigg|.
\end{equation}
Note that Lemma~\ref{lem:alpha-expansion} also implies that $\max_i|\widehat{\alpha}_i-\alpha|^2 =  O_p(T^{-3/4})$.

Next we bound the second term in~\eqref{eq:iv_bound}. First, we derive an expression for $U_{ir\alpha\alpha}$ using the derivation of $U_{ir\alpha}$ above:
\[
  \begin{aligned}
  U_{ir\alpha\alpha}(\gamma_{0}, \overline{\beta}, \overline{\alpha}_{i})
  & = U_{ir\alpha}(\gamma_{0}, \overline{\beta}, \overline{\alpha}_{i})\overline{V}_i(\overline{\beta}, \overline{\alpha}) +
  U_{ir}(\gamma_{0}, \overline{\beta}, \overline{\alpha}_{i})  \frac{\partial}{\partial \alpha} \overline{V}_i(\overline{\beta}, \overline{\alpha}) \\
  &= U_{ir}(\gamma_{0}, \overline{\beta}, \overline{\alpha}_{i})\overline{V}_{i\alpha}(\overline{\beta}, \overline{\alpha}),
  \end{aligned}
\]
where we define $\overline{V}_{i\alpha}(\beta,\alpha) = \overline{V}_i(\beta, \alpha)^2 + \partial \overline{V}_i(\beta, \alpha)/\partial \alpha$.

With this, we bound the second moment of the second term in a neighborhood around the truth:
\begin{align*}
\E\left[\sup_{(\alpha, \beta) \in B_{0}(\epsilon)}
\left( \frac{1}{N} \sum_{i=1}^N U_{ir\alpha\alpha}(\gamma_{0}, \beta, \alpha_{i})  \right)^2
\right]
&  = \E\left[\sup_{(\alpha, \beta) \in B_{0}(\epsilon)}
\left( \frac{1}{N} \sum_{i=1}^N U_{ir}(\gamma_{0}, \beta, \alpha_{i})\overline{V}_{i\alpha}(\beta, \alpha) \right)^2
\right] \\
&  \leq \E\left[\sup_{(\alpha, \beta) \in B_{0}(\epsilon)}
\left( \frac{1}{N} \sum_{i=1}^N |U_{ir}(\gamma_{0}, \beta, \alpha_{i})||\overline{V}_{i\alpha}(\beta, \alpha)| \right)^2
\right] \\
&  \leq \E\left[\sup_{(\alpha, \beta) \in B_{0}(\epsilon)}
\left( \frac{1}{N} \sum_{i=1}^N |Y_i - g(\underline{D}_{ik};\widehat{\gamma})||\overline{V}_{i\alpha}(\beta, \alpha)| \right)^2
\right] \\
&  \leq \E\left[
\left( \frac{1}{N} \sum_{i=1}^N |Y_i - g(\underline{D}_{ik};\widehat{\gamma})|M_i \right)^2
\right] \\
&  \leq  \frac{1}{N} \sum_{i=1}^N \E\left[
\left( |Y_i - g(\underline{D}_{ik};\widehat{\gamma})|M_i \right)^2
\right] \\
&\leq \frac{1}{N} \sum_{i=1}^N \left( \E[(Y_i - g(\underline{D}_{ik};\widehat{\gamma}))^4] \right)^{1/2}\left( \E[M_i^4] \right)^{1/2} = O(1)
\end{align*}
The second inequality here is due to bounded propensity scores, the third due to Lemma~\ref{lem:msm-second-order}, the fourth due to i.i.d. data, the fifth is Cauchy-Swarchz, and the final equality is due to bounded outcome moments and Lemma~\ref{lem:msm-second-order}. This implies that
\[
\sup_{(\alpha, \beta) \in B_{0}(\epsilon)}\bigg|\frac{1}{N}\sum^{N}_{i=1}U_{ir\alpha\alpha}(\gamma_{0}, \beta, \alpha_{i})\bigg| = O_p(1).
\]
Note that $(\overline{\beta}, \overline{\alpha}_i) \in \mathcal{B}(\epsilon)$ with probability approaching 1 due to these values between the consistent estimators and the true values of the parameters. This implies that
\[
\bigg|\frac{1}{N}\sum^{N}_{i=1}U_{ir\alpha\alpha}(\gamma_{0}, \overline{\beta}, \overline{\alpha}_{i})\bigg| = O_p(1).
\]
Combining this with the above, we have that (IV) is $\sqrt{N}O_p(T^{-3/4})O_p(1) = o_p(1)$.

For (V), we follow a similar strategy. First note that we have:
\[
\abs*{\text{(V)}} < \sqrt{N}\norm{\widehat{\beta} - \beta_0}\max_i|\widehat{\alpha}_{i} - \alpha_{i}|\Bigg\Vert\frac{1}{N} \sum_{i=1}^N U_{ir\beta\alpha}(\gamma_{0}, \overline{\beta}, \overline{\alpha}_{i})\Bigg\Vert
\]
As above, $\sqrt{N}\norm{\widehat{\beta} - \beta_0}\max_i|\widehat{\alpha}_{i} - \alpha_{i}| = \sqrt{N}O_p(1/\sqrt{NT})O_p(T^{-3/8}) = O_p(T^{-7/8})$. Let $U_{irq\alpha} = \partial U_{ir\alpha}/\partial \beta_q$ be the $q$th entry of $U_{ir\beta\alpha}$.  By a similar argument to $U_{ir\alpha\alpha}$, we have
\[
U_{irq\alpha}(\beta, \alpha) = U_{ir}(\beta, \alpha)\left( \overline{V}_i(\beta, \alpha) \overline{S}_{iq}(\beta, \alpha) + \partial \overline{V}_i(\beta, \alpha)/\partial \beta_q\right) \equiv U_{ir}(\beta, \alpha)\overline{V}_{iq}(\beta, \alpha).
\]
By Lemma~\ref{lem:msm-second-order} and the argument for $U_{ir\alpha\alpha}$, we have
\[
\abs*{\frac{1}{N} \sum_{i=1}^N U_{irq\alpha}(\gamma_{0}, \overline{\beta}, \overline{\alpha}_{i})} = O_p(1),
\]
for all $q$ which in turn implies,
\[
\Bigg\Vert\frac{1}{N} \sum_{i=1}^N U_{ir\beta\alpha}(\gamma_{0}, \overline{\beta}, \overline{\alpha}_{i})\Bigg\Vert = O_p(1)
\]
Thus, we have that (V) is $O_p(T^{-7/8})O_p(1) = o_p(1)$. The proof for (III) being $o_p(1)$ follows similarly.


\paragraph{Combining all results}

Combining all results, we have that
\begin{align*}
\sqrt{N}(\widehat{\gamma} - \gamma_{0})
&= G^{-1} \frac{1}{\sqrt{N}}\sum^{N}_{i=1}U_{i} + o_{p}(1)
\end{align*}

\end{proof}

%% file: psfe_lemmas.tex



\subsection{Lemmas}

The following lemma is a restatement of results in \cite{FerWei16} with an additional result on the uniform rate of converge for the propensity score model.
\begin{lemma}\label{lem:alpha-expansion}
  Under Assumption~\ref{a:panel-data}, the following hold:
  \begin{enumerate}
    \item[(i)] $\Vert\widehat{\beta} - \beta_0\Vert = O_p(1/\sqrt{NT})$
  \item[(ii)] Letting $\psi_{it} = \E_T\{\E_\alpha[V_{it\alpha}]\}^{-1}V_{it}$, we have:

    \[
      \widehat{\alpha}_i  = \alpha_{i0} + \frac{1}{T}\sum_{t=1}^T \psi_{it} + R_i,
    \]
    where $R_i = O_p(1/T)$.
    \item[(iii)] $\max_i |\widehat{\alpha}_i - \alpha_{i0}| = \max_{i}|T^{-1}\sum_{t=1}^T \psi_{it}| = O_p(T^{-3/8})$ and $\max_i|R_i| = o_p(T^{-1/2})$.
  \end{enumerate}
\end{lemma}

The following can be found as Proposition 2.5 of~\citet{FanYao05}.

\begin{lemma}\label{lem:mixing-covariance}
  Let $\{\xi_t\}$ be an $\alpha$-mixing process with mixing coefficient $a(m)$. Let $\E|\xi_t|^p < \infty$ and $\E|\xi_{t+m}|^q < \infty$ for some $p, q \geq 1$ and $1/p + 1/q < 1$. Then,
  \[
    |\textrm{Cov}(\xi_t, \xi_{t+m})| \leq 8 a(m)^{1/r}\left[ \E|\xi_t|^p \right]^{1/p}\left[ \E|\xi_{t+m}|^q \right]^{1/q},
  \]
where $r = (1 - 1/p - 1/q)^{-1}$.
\end{lemma}

The following lemma comes from Theorem 1 of~\citet{CoxKim95}.

\begin{lemma}\label{lem:mixing-bounds}
  Let $\{\xi_t\}$ be an $\alpha$-mixing process with mixing coefficient $a(m)$ and $\E[\xi_t] = 0$. Let $r\ge 1$ be an integer, and let $\delta > 2r$, $\mu > r / (1 - 2r / \delta)$, $c > 0$, and $C > 0$. Assume that $\sup_t \E[|\xi_t|^{\delta}] \leq C$ and that $a(m) \le cm^{-\mu}$ for all $m \in \{1,2,3,\ldots\}$. Then there exists a constant $B > 0$ depending on $r, \delta, \mu, c$ and $C$, but not depending on $T$ or any other distributional characteristics of $\xi_t$ such that for any $T > 0$,
  \[
    \E\left[ \left( \frac{1}{\sqrt{T}} \sum_{t=1}^T \xi_t \right)^{2r} \right] \leq B.
  \]
\end{lemma}

\begin{lemma}\label{lem:alpha-remainder}
Let $\ell_i^*(\beta, V)$ be the Legendre transformation of the objective function $\ell_{i}(\beta, \alpha) = T^{-1}\sum_{t=1}^T \ell_{it}(\beta, \alpha)$ such that 
$\ell_{i}^*(\beta, V) = \max_{\alpha \in \mathcal{B}_{\alpha}(\epsilon_{\alpha})} [\ell_{i}(\beta, \alpha) - \alpha V]$
and $A_i(\beta, V) = \argmax_{\alpha \in \mathcal{B}_{\alpha}(\epsilon_{\alpha})} [\ell_{i}(\beta, \alpha) - \alpha V]$
where $\beta \in \mathcal{B}_{\beta}(\epsilon_{\beta})$ and $V$ denotes the dual parameter to $\alpha$. Suppose Assumption~\ref{a:panel-data} holds. Then, for some $\nu > 0$
  \begin{enumerate}
  \item[(i)]
      \[
    \begin{aligned}
      \sup_{(\beta,\alpha_i) \in \mathcal{B}_{0}(\epsilon)} &|\partial_{VVVV} \ell^*_i(\beta, \alpha_i)| = O_p(1)\\
      \sup_{(\beta,\alpha_i) \in \mathcal{B}_{0}(\epsilon)} &\Vert \partial_{V\beta\beta'} \ell_i^*(\beta, \alpha_i)\Vert = O_p(1)\\
      \sup_{(\beta,\alpha_i) \in \mathcal{B}_{0}(\epsilon)} & \Vert \partial_{VV\beta} \ell_i^*(\beta, \alpha_i)\Vert = O_p(1)\\
    \end{aligned}
  \]
  \item[(ii)]
    \[
      \begin{aligned}
      \sup_{(\beta,\alpha_i) \in \mathcal{B}_{0}(\epsilon)} &\max_i |\partial_{VVVV} \ell^*_i(\beta, \alpha_i)| = O_p(T^{1/(8+\nu)})\\
      \sup_{(\beta,\alpha_i) \in \mathcal{B}_{0}(\epsilon)} &\max_i \Vert \partial_{V\beta\beta'} \ell_i^*(\beta, \alpha_i)\Vert = O_p(T^{2/(8+\nu)})\\
      \sup_{(\beta,\alpha_i) \in \mathcal{B}_{0}(\epsilon)} &\max_i \Vert \partial_{VV\beta} \ell_i^*(\beta, \alpha_i)\Vert = O_p(T^{2/(8+\nu)})\\
    \end{aligned}
  \]
  \end{enumerate}
\end{lemma}

\begin{lemma}\label{lem:msm-second-order}
Let $V_{it}(\beta, \alpha) = \partial \ell_{it}(\beta, \alpha)/ \alpha$ and $S_{itq}(\beta, \alpha) = \partial \ell_{it}(\beta, \alpha) / \partial \beta_q$. Define the following:
  \[
    \begin{aligned}
      \overline{V}_i(\beta,\alpha) &= \sum^{T}_{t=T-k} \bigg[
  (2D_{it}-1) V_{it}(\beta,\alpha) \bigg\{ \frac{\pi_{it}(\beta,\alpha)}{1 - \pi_{it}(\beta,\alpha)} \bigg\}^{1-D_{it}}
  \bigg] \\
  \overline{S}_{iq}(\beta,\alpha) & = \sum^{T}_{t=T-k} \bigg[
  (2D_{it}-1) S_{itq}(\beta,\alpha) \bigg\{ \frac{\pi_{it}(\beta,\alpha)}{1 - \pi_{it}(\beta,\alpha)} \bigg\}^{1-D_{it}}
\bigg] \\
      \overline{V}_{i\alpha}(\beta,\alpha) &= \overline{V}_i(\beta, \alpha)^2 + \partial \overline{V}_i(\beta, \alpha)/\partial \alpha, \\
      \overline{V}_{iq}(\beta,\alpha) &= \overline{V}_i(\beta, \alpha)\overline{S}_{iq}(\beta, \alpha) + \partial \overline{V}_i(\beta, \alpha)/\partial \beta_q, \\
      \overline{S}_{iq\alpha}(\beta,\alpha) &= \overline{V}_i(\beta, \alpha)\overline{S}_{iq}(\beta, \alpha) + \partial \overline{S}_i(\beta, \alpha)/\partial \alpha, \\
      \overline{S}_{iqm}(\beta,\alpha) &= \overline{S}_{iq}(\beta, \alpha)\overline{S}_{im}(\beta, \alpha) + \partial \overline{S}_{iq}(\beta, \alpha)/\partial \beta_m \\
    \end{aligned}
  \]
Suppose Assumption~\ref{a:panel-data} holds. Then each of these is uniformly bounded in absolute value for $(\beta,\alpha) \in \mathcal{B}_{0}(\epsilon)$ by a function $\widetilde{M}_i$ such that $\max_{i,t} \E[\widetilde{M}_i^{4}]$ is almost surely uniformly bounded over $N$.
\end{lemma}

\subsection{Proof of Lemmas}

\begin{proof}[Proof of Lemma~\ref{lem:alpha-expansion}]
  We take the convention here that any function with the $t$ subscript omitted is the over-time average of that quantity. For example, $V_i(\beta, \alpha) = \E_T[V_{it}(\beta, \alpha)] = T^{-1} \sum_{t=1}^T V_i(\beta, \alpha)$.


Part (i) follows from the results of \citet{FerWei16} without period effects.

To derive an asymptotic expansion of $\widehat{\alpha}_i$ in part (ii), we largely follow the Legendre transformation approach of \cite{FerWei16}. Our discussion follows theirs closely, though in a more specialized setting. We define
\[
\ell_{i}^*(\beta, V) = \max_{\alpha \in \mathcal{B}_{\alpha}(\epsilon_{\alpha})} [\ell_{i}(\beta, \alpha) - \alpha V], \qquad A_i(\beta, V) = \argmax_{\alpha \in \mathcal{B}_{\alpha}(\epsilon_{\alpha})} [\ell_{i}(\beta, \alpha) - \alpha V]
\]
where $\beta \in \mathcal{B}_{\beta}(\epsilon_{\beta})$. The function $\ell_i^*(\beta, V)$ is the Legendre transformation of the objective function $\ell_{i}(\beta, \alpha) = T^{-1}\sum_{t=1}^T \ell_{it}(\beta, \alpha)$. We use $V$ to denote the dual parameter to $\alpha$ and $\ell_i^*(\beta, V)$ as the dual function to $\ell_i(\beta, \alpha)$. The relationship between $\alpha$ and $V$ is one-to-one since the optimal $\alpha = A_i(\beta, V)$ satisfies the first-order condition $V_i(\beta, \alpha) = V$, where $V_i(\beta, \alpha) = T^{-1}\sum_{t=1}^TV_{it}(\beta, \alpha)$. We can write $\ell_i^*(\beta, V) = \ell_i(\beta, A_i(\beta, V)) - A_i(\beta, V)V$ when $A_i(\beta, V)$ solves the FOC, $V_i(\beta, A_i(\beta, V_i)) = V_i$. Taking the derivative of the last identity on both sides of the equality gives:
\[
  \begin{aligned}
    [\partial_{V} A_{i}(\beta, V)][V_{i\alpha}(\beta, A_i(\beta, V))] &= 1 \\
    V_{i\beta}(\beta, A_i(\beta, V)) + [\partial_{\beta}A_i(\beta, V)]V_{i\alpha}(\beta, A_i(\beta, V)) &= 0
  \end{aligned}
\]
so we have
\[
  \begin{aligned}
    \partial_{V} A_{i}(\beta, V) &= \frac{1}{V_{i\alpha}(\beta, A_i(\beta, V))} \\
    \partial_{\beta}A_i(\beta, V) &= -\frac{V_{i\beta}(\beta, A_i(\beta, V))}{V_{i\alpha}(\beta, A_i(\beta, V))}
  \end{aligned}
\]

When $V = 0$, then the optimization in $\ell_i^*$ is just over $\ell_i$, so we have $A_i(\beta, 0) = \widehat{\alpha}_i(\beta)$ and $\ell_i^*(\beta, 0) = \ell_i(\beta, \widehat{\alpha}_i(\beta))$. The latter is the profile likelihood for $\beta$. Note that
\[
  \begin{aligned}
    \partial_{V}\ell^*_i(\beta, A_i(\beta, V)) &= [V_{i}(\beta, A_i(\beta, V))] [\partial_{V} A_i(\beta, V)] - A_i(\beta, V) - [\partial_{V} A_i(\beta, V)]V \\
    &= - A_i(\beta, V)
  \end{aligned}
\]
Thus, we have $\widehat{\alpha}_i(\beta) = -\partial_V \ell^*_i(\beta, 0)$ and $\partial_V \ell^*_i(\beta_0, 0) = -\alpha_{i0}$.

We now expand $\partial_V \ell^*_i(\beta, 0)$ around $(\beta_0, V_i)$ for $\beta \in \mathcal{B}_{\beta_0}(\epsilon_{\beta})$, which gives:
\[
\widehat{\alpha}_i(\beta) = -\partial_V \ell^*_i(\beta, 0) = -\partial_V \ell^*_i - (\partial_{V\beta'}\ell_i^*)(\beta - \beta_0) + (\partial_{VV}\ell_i^*)V_i - \frac{1}{2}(\partial_{VVV} \ell_i^*)V_i^2 + R(\beta),
\]
where
\[
  \begin{aligned}
    R_i(\beta) &=
    \frac{1}{2} (\beta - \beta_0)^{\top}(\partial_{V\beta\beta'} \ell_i^*(\overline{\beta}, V_i))(\beta - \beta_0)
    + (\partial_{VV\beta'} \ell_i^*(\beta_0, \widetilde{V}_i))(\beta - \beta_0)V_i \\
    &\qquad + \frac{1}{6} (\partial_{VVVV} \ell_i^*(\beta_0, \ddot{V}_i))V_i^3
  \end{aligned}
\]
where $\overline{\beta}$ is between $\beta$ and $\beta_0$ and $\widetilde{V}_i$ and $\ddot{V}_i$ are between $V_i$ and 0.

Using the above identities, it is possible to derive the following:
\[
  \begin{aligned}
    \partial_{VV} \ell_i^* &= -1/V_{i\alpha}, \qquad  \partial_{VVV} \ell_i^*= V_{i\alpha\alpha}/V^3_{i\alpha}, \qquad \partial_{\beta}\ell_i^* = S_i \\
    \partial_{V\beta} \ell_i^* &= V_{i\beta} / V_{i\alpha}, \qquad \partial_{\beta\beta'} \ell_i^* = S_{i\beta} + V_{i\beta}V^{\top}_{i\beta}/V_{i\alpha} \\
    \partial_{V\beta\beta'} \ell_i^*
    &=\frac{S_{i\beta\alpha}}{V_{i\alpha}} + \frac{V_{i\beta\alpha}V_{i\beta}^\top}{V_{i\alpha}^2} + \frac{V_{i\beta}V_{i\beta\alpha}^\top}{V_{i\alpha}^2} - \frac{V_{i\beta}V_{i\beta}^\top}{V_{i\alpha}^3}\\
    \partial_{VV\beta} \ell_i^*
    &= \frac{V_{i\beta\alpha}}{V_{i\alpha}^2} - \frac{V_{i\alpha\alpha}V_{i\beta}}{V_{i\alpha}^3} \\
    \partial_{VVVV} \ell_{i}^* &= \frac{V_{i\alpha\alpha\alpha}}{V_{i\alpha}^4} - 3\frac{V_{i\alpha\alpha}}{V_{i\alpha}^5}
  \end{aligned}
\]
Combining these with the above expansion gives us:
\[
\widehat{\alpha}_i(\beta) - \alpha_{i0}
= - \frac{V_i}{V_{i\alpha}} - \frac{V_{i\beta}^{\top}(\beta - \beta_0)}{V_{i\alpha}}  + \frac{V_{i\alpha\alpha}V_i^2}{2V_{i\alpha}^3} + R(\beta)
\]

Let  $\E_{\alpha}$ is the expectation conditional on the fixed effects. We have $\E_{\alpha}[V_i] = 0$ and so by Lemma~\ref{lem:mixing-bounds},  we have $V_i = O_p(T^{-1/2})$, $V_i^2 = O_p(T^{-1})$, and $V_i^3 = O_p(T^{-3/2})$.
For the partial derivative terms, we define $\overline{V}_{it\alpha} = \E_{\alpha}[V_{it\alpha}]$ and $\overline{V}_{i\alpha} = \E_T[\overline{V}_{it\alpha}]$. We also define the mean deviations as $\widetilde{V}_{i\alpha} = V_{i\alpha} - \overline{V}_{i\alpha}$. We define similar quantities for $V_{i\alpha\alpha}$ and $V_{i\beta}$.  As above, we have $\widetilde{V}_{i\alpha} = O_p(T^{-1/2})$. We can derive the following using standard asymptotic results:
\[
  \begin{aligned}
    V_{i\alpha}^{-1} &= \overline{V}_{i\alpha}^{-1} + \overline{V}_{i\alpha}^{-2}\widetilde{V}_{i\alpha} + O_p(T^{-1}) \\
    V_{i\alpha}^{-2} &= \overline{V}_{i\alpha}^{-2} + \overline{V}_{i\alpha}^{-3}\widetilde{V}_{i\alpha} + O_p(T^{-1}) \\
    V_{i\alpha}^{-3} &= \overline{V}_{i\alpha}^{-3} + \overline{V}_{i\alpha}^{-4}\widetilde{V}_{i\alpha} + O_p(T^{-1})
  \end{aligned}
\]

Let $\psi_{i1} = V_{i}/\overline{V}_{i\alpha}$. We have $V_i/V_{i\alpha} = \psi_{i1} + \psi_{i1}(\widetilde{V}_{i\alpha}/\overline{V}_{i\alpha}) + O_p(T^{-3/2})$. For the next two terms, we have
\[
  \begin{aligned}
    \frac{V_{i\beta}^\top(\beta - \beta_0)}{V_{i\alpha}}
    &= \frac{V_{i\beta}^\top(\beta - \beta_0)}{\overline{V}_{i\alpha}} + O_P(T^{-1/2}\Vert \beta - \beta_0\Vert) \\
    \frac{V_{i\alpha\alpha}V_i^2}{2V_{i\alpha}^3} &= \frac{V_{i\alpha\alpha}V_{i}^2}{2\overline{V}^3_{i\alpha}} + O_P(T^{-3/2}) = \frac{V_{i\alpha\alpha}\psi_{i1}^2}{2\overline{V}_{i\alpha}} + O_P(T^{-3/2})
  \end{aligned}
\]

Again based on the bounded derivative and moment conditions and Lemma~\ref{lem:mixing-bounds}, we have  $\widetilde{V}_{i\alpha\alpha} = O_P(T^{-1/2})$ and $\Vert\widetilde{V}_{i\beta}\Vert = O_p(T^{-1/2})$. Combined with the above, we have:

\[
\widehat{\alpha}_i(\beta) - \alpha_{i0}
= -\psi_{i1} - \psi_{i2} + R_{i}(\beta) + O_p\left(T^{-1/2}\Vert \beta - \beta_0\Vert + T^{-3/2}\right),
\]
where
\[
  \psi_{i2} = \frac{1}{\overline{V}_{i\alpha}}\left(\widetilde{V}_{i\alpha}\psi_{i1} + \overline{V}_{i\beta}^\top (\beta - \beta_0) + \overline{V}_{i\alpha\alpha}\psi_{i1}^2 \right),
\]
and $\psi_{i2} = O_p(T^{-1})$.

We now plug in $\widehat{\beta}$ to get $\widehat{\alpha}_i = \widehat{\alpha}_i(\widehat{\beta})$. Note that $\Vert\widehat{\beta}-\beta_0\Vert = O_p((NT)^{-1/2}) = O_p(T^{-1})$. Now by part \emph{(i)} of Lemma~\ref{lem:alpha-remainder} we have
\[
  \begin{aligned}
    |R_i(\widehat{\beta})|
    & \leq \Vert\widehat{\beta}-\beta_0\Vert^2\Vert{\partial_{V\beta\beta'} \ell_i^*(\overline{\beta}, V_i)}\Vert
    + \abs*{V_i}\;\Vert\widehat{\beta}-\beta_0\Vert\; \Vert\partial_{VV\beta'} \ell_i^*(\beta_0, \widetilde{V}_i)\Vert \\
    &\qquad + \abs*{V_i^3}\;\abs*{\partial_{VVVV} \ell_i^*(\beta_0, \ddot{V}_i)} \\
    &= O_P(T^{-2})O_p(1)
    + O_{p}(T^{-1/2})O_p(T^{-1})O_p(1) + O_p(T^{-3/2})O_p(1) = O_P(T^{-3/2})
  \end{aligned}
\]
Combining this with the above, we have
\[
\widehat{\alpha}_i - \alpha_i = -\psi_{i1} - \psi_{i2} + O_p(T^{-3/2})
\]

For part (iii), we now derive a maximal inequality over units. We have
\[
\max_{i} \abs*{\widehat{\alpha}_i - \alpha_{i0}} \leq \max_i \abs*{\psi_{i1}} + \max_i \abs*{\psi_{i2}} + \max_i \abs*{R_{i}(\beta)}
\]
By Lemma~\ref{lem:mixing-bounds}, we have:
\[
  \begin{aligned}
    \E_{\alpha}\left[ \left( \max_i \abs*{V_i} \right)^8 \right] &= T^{-4}\E_{\alpha}\left[\max_i \left( \frac{1}{\sqrt{T}} \sum_{t=1}^T V_{it} \right)^8\right] \\
    &\leq T^{-4} \sum_{i} \E_{\alpha}\left[ \left( \frac{1}{\sqrt{T}} \sum_{t=1}^T V_{it} \right)^8 \right] \\
    &\leq T^{-4}NB = O(T^{-3})
  \end{aligned}
\]
Thus, we have $\max_i |V_i| = O_p(T^{-3/8})$, $\max_{i} |V_i^2| = O_p(T^{-3/16})$, and $\max_i |V_i^3| = O_p(T^{-9/8})$.  Furthermore, recall that $\widehat{\alpha}_i = \widehat{\alpha}_i(\widehat{\beta})$ and that $\Vert\widehat{\beta} - \beta_0\Vert$ is $O_p((NT)^{-1/2}) = O_p(T^{-1})$. Finally, recall that $\inf_i |V_{it\alpha}| > 0$, which implies that $\inf_i |\overline{V}_{i\alpha}| > 0$.  These facts combined with part \emph{(ii)} of Lemma~\ref{lem:alpha-remainder} imply for some constants $C_1$
\[
  \begin{aligned}
    \max_i \abs*{\psi_{i1}} &= \max_i \abs*{\overline{V}_{i\alpha}^{-1}V_i} < C\max_i \abs*{V_i} =  O_p(T^{-3/8}) \\
    \max_i\abs*{\psi_{i2}} &= \max_i \abs*{\overline{V}_i^{-1}\left(\overline{V}_{i\beta}(\widehat{\beta} - \beta_0) + \overline{V}_{i\alpha\alpha}\psi_{i1}^2 \right)} \\
    &< C \left(  \Vert\widehat{\beta} - \beta_0\Vert \max_i\Vert\overline{V}_{i\beta}\Vert + \max_i\abs*{\overline{V}_{i\alpha\alpha}}\max_i\abs*{\psi_{i1}^2}\right) \\
    & = O_p(T^{-1} + T^{-3/4}) = o_p(T^{-1/2}),\\
    \max_i |R_i(\widehat{\beta})| &= o_p(T^{-1}),
  \end{aligned}
\]
Here, we used the fact that $\max_i |\overline{V}_{i\alpha\alpha}| < E_{\alpha}[M(Z_{it})]$, which is uniformly bounded over $i$ and $t$.  These three combined implies $\max_i \abs*{\widehat{\alpha}_i - \alpha_{i0}} = O_p(T^{-3/8})$.
\end{proof}

\begin{proof}[Proof of Lemma~\ref{lem:alpha-remainder}]
  Part \emph{(i)}. Let $\xi_{it}$ be one of $V_{it\beta_k}$, $V_{it\alpha\alpha}$, $V_{it\alpha\alpha\alpha}$, $V_{i\beta_k\alpha}$, or $V_{it\beta_k\beta_j\alpha}$ and note that $E[\abs*{\xi_{it}}^{8+\nu}] < E[M(Z_{it})^{8+\nu}]< \infty$ by Assumption~\ref{a:bounded-derivatives}. We have:
  \[
  \begin{aligned}
    \E_{\alpha}\left[ \left( \sup_{(\beta, \alpha) \in \mathcal{B}_0(\epsilon)} \frac{1}{T}\sum_{t=1}^T |\xi_{it}(\beta, \alpha)| \right)^{(8+\nu)} \right]
    &\leq \E_{\alpha}\left[\left(  \frac{1}{T}\sum_{t=1}^T M(Z_{it}) \right)^{(8+\nu)} \right]\\
    &\leq \E_{\alpha}\left[\frac{1}{T}\sum_{t=1}^T M(Z_{it})^{(8+\nu)} \right]\\
    & = \frac{1}{T} \sum_{t=1}^T \E_{\alpha}[M(Z_{it})^{(8+\nu)}] = O_p(1)
  \end{aligned}
\]
Thus, $\sup_{(\beta, \alpha) \in \mathcal{B}_0(\epsilon)} \frac{1}{T}\sum_{t=1}^T |\xi_{it}(\beta, \alpha)| = O_p(1)$. From the expression for $\partial_{VV\beta}$ given above, we have
\[
  \begin{aligned}
    \sup_{(\beta, \alpha) \in \mathcal{B}_0(\epsilon)} \abs*{\partial_{VV\beta_k} \ell_i^*(\beta, \alpha)}
    &\leq \sup_{(\beta, \alpha) \in \mathcal{B}_0(\epsilon)} \left(   \abs*{\frac{V_{i\beta_k\alpha}(\beta, \alpha)}{V_{i\alpha}(\beta, \alpha)^2}} + \abs*{\frac{V_{i\alpha\alpha}(\beta, \alpha)V_{i\beta_k}(\beta, \alpha)}{V_{i\alpha}^3(\beta, \alpha)}}\right) \\
    & < \sup_{(\beta, \alpha) \in \mathcal{B}_0(\epsilon)} \left(\abs*{V_{i\beta_k\alpha}(\beta, \alpha)} + \abs*{V_{i\alpha\alpha}(\beta, \alpha)V_{i\beta_k}(\beta, \alpha)}\right) \\
    &= O_p(1)
  \end{aligned}
\]
The second inequality follows from $\inf_i |V_{it\alpha}(\beta, \alpha)| > 0$. The other statements in part \emph{(i)} follow analogously.

For the maximal results in part \emph{(ii)}, we follow a similar strategy:
  \begin{align*}
    \E\left[ \left( \sup_{(\beta, \alpha) \in \mathcal{B}_0(\epsilon)} \max_i \frac{1}{T}\sum_{t=1}^T |\xi_{it}(\beta, \alpha)| \right)^{(8+\nu)} \right]
    &= \E\left[\max_i \left( \sup_{(\beta, \alpha) \in \mathcal{B}_0(\epsilon)}  \frac{1}{T}\sum_{t=1}^T |\xi_{it}(\beta, \alpha)| \right)^{(8+\nu)} \right] \\
    &\leq \E\left[\sum_{i=1}^N \left( \sup_{(\beta, \alpha) \in \mathcal{B}_0(\epsilon)}  \frac{1}{T}\sum_{t=1}^T |\xi_{it}(\beta, \alpha)| \right)^{(8+\nu)} \right] \\
    &\leq \E\left[\sum_{i=1}^N \left(  \frac{1}{T}\sum_{t=1}^T M(Z_{it}) \right)^{(8+\nu)} \right]\\
    &\leq \E\left[\sum_{i=1}^N   \frac{1}{T}\sum_{t=1}^T M(Z_{it})^{(8+\nu)} \right] \\
    &= \sum_{i=1}^N \frac{1}{T} \sum_{t=1}^T \E[M(Z_{it})^{(8+\nu)}] = O(N)
  \end{align*}
Thus, $\sup_{(\beta, \alpha) \in \mathcal{B}_0(\epsilon)} \max_i \frac{1}{T}\sum_{t=1}^T |\xi_{it}(\beta, \alpha)| = O_p(N^{1/(8+\nu)}) = O_p(T^{1/(8+\nu)})$. From above we have,
\[
  \begin{aligned}
    \sup_{(\beta, \alpha) \in \mathcal{B}_0(\epsilon)}& \max_i \abs*{\partial_{VV\beta_k} \ell_i^*(\beta, \alpha)} \\
    &\leq \sup_{(\beta, \alpha) \in \mathcal{B}_0(\epsilon)} \max_i\left(   \abs*{\frac{V_{i\beta_k\alpha}(\beta, \alpha)}{V_{i\alpha}(\beta, \alpha)^2}} + \abs*{\frac{V_{i\alpha\alpha}(\beta, \alpha)V_{i\beta_k}(\beta, \alpha)}{V_{i\alpha}^3(\beta, \alpha)}}\right) \\
    & < \sup_{(\beta, \alpha) \in \mathcal{B}_0(\epsilon)} \max_i\left(\abs*{V_{i\beta_k\alpha}(\beta, \alpha)} + \abs*{V_{i\alpha\alpha}(\beta, \alpha)V_{i\beta_k}(\beta, \alpha)}\right) \\
    &= O_p\left( T^{1/(8+\nu)} + T^{2/(8+\nu)} \right) = O_p(T^{2/(8+\nu)})
  \end{aligned}
\]
The other results follow analogously.
\end{proof}

\begin{proof}[Proof of Lemma~\ref{lem:msm-second-order}]
  We prove this for $\overline{V}_i(\beta, \alpha)$ and $\overline{V}_{i\alpha}(\beta, \alpha)$, the rest follow from very similar arguments. Let $C_{\pi}$ be a uniform bound on $(\pi_{it}(\beta, \alpha)/(1-\pi_{it}(\beta, \alpha)))$ in $(\beta,\alpha) \in \mathcal{B}_{0}(\epsilon)$,
  which exists by the virtue of bounded propensity scores (Assumption~\ref{a:bounded-propensity-scores}). Furthermore, let $M_i = \max_t M(Z_{it})$ be a uniform bound for all $V_{it}(\beta, \alpha)$, which exists due to Assumption~\ref{a:bounded-derivatives}. Then, within $\mathcal{B}_{0}(\epsilon)$, for $\overline{V}_i$  we have:
  \[
    \begin{aligned}
      \abs*{\overline{V}_i(\beta, \alpha)} &= \Bigg\vert\sum^{T}_{t=T-k} \bigg[
  (2D_{it}-1) V_{it}(\beta,\alpha) \bigg\{ \frac{\pi_{it}(\beta,\alpha)}{1 - \pi_{it}(\beta,\alpha)} \bigg\}^{1-D_{it}}
  \bigg] \Bigg\vert\\
  & \leq \sum^{T}_{t=T-k}
  \vert V_{it}(\beta,\alpha)\vert \bigg\{ \frac{\pi_{it}(\beta,\alpha)}{1 - \pi_{it}(\beta,\alpha)} \bigg\}^{1-D_{it}}   \\
  & < C_{\pi}\sum^{T}_{t=T-k}
  \vert V_{it}(\beta,\alpha)\vert < C_{\pi}kM_i
    \end{aligned}
  \]
  Thus, $\sup_{(\beta,\alpha)\in\mathcal{B}_{0}(\epsilon)} \abs*{\overline{V}_i(\beta, \alpha)} < \widetilde{M}_i = C_{\pi}kM_i$. Since $\E[\widetilde{M}_i^4] = C_{\pi}^4k^4\E[M_i^4]$, we obtain the result for the first quantity.

  For $\overline{V}_{i\alpha}$ there are two terms. For the first term, we have
    \[
    \begin{aligned}
      \abs*{\overline{V}_i^2(\beta, \alpha)} &= \Bigg\vert\sum_{t=T-k}^T \sum_{s=T-k}^T (2D_{it} - 1)(2D_{is} - 1)V_{it}(\beta, \alpha)V_{is}(\beta, \alpha) \\
      &\qquad \times {\left(  \frac{\pi_{it}(\beta, \alpha)}{1-\pi_{it}(\beta, \alpha)}\right)}^{1-D_{it}}{\left(  \frac{\pi_{is}(\beta, \alpha)}{1-\pi_{is}(\beta, \alpha)}\right)}^{1-D_{is}} \Bigg\vert \\
      &\leq \sum_{t=T-k}^T \sum_{s=T-k}^T \abs*{V_{it}(\beta, \alpha)V_{is}(\beta, \alpha)}{\left(  \frac{\pi_{it}(\beta, \alpha)}{1-\pi_{it}(\beta, \alpha)}\right)}^{1-D_{it}}{\left(  \frac{\pi_{is}(\beta, \alpha)}{1-\pi_{is}(\beta, \alpha)}\right)}^{1-D_{is}} \\
    &< C_{\pi}^2  \sum_{t=T-k}^T \sum_{s=T-k}^T \abs*{V_{it}(\beta, \alpha)V_{is}(\beta, \alpha)} \\
    &< C_{\pi}^2k^2M^2_i
    \end{aligned}
  \]

  Now, for the second term, we have:
  \[
    \begin{aligned}
      \frac{\partial \overline{V}_i(\beta, \alpha)}{\partial \alpha} & = \sum_{t=T-k}^T (2D_{it} - 1) V_{it\alpha}(\beta, \alpha){\left(  \frac{\pi_{it}(\beta, \alpha)}{1-\pi_{it}(\beta, \alpha)}\right)}^{1-D_{it}} \\
      &\qquad + \sum_{t=T-k}^T \frac{(1 - D_{it})}{1-\pi_{it}(\beta, \alpha)} V_{it}(\beta, \alpha)^2
      {\left(  \frac{\pi_{it}(\beta, \alpha)}{1-\pi_{it}(\beta, \alpha)}\right)}^{1-D_{it}}
    \end{aligned}
  \]

  Let $\overline{C}_{\pi} > (1 - \pi_{it}(\beta, \alpha))^{-1}$, which exists by Assumption~\ref{a:bounded-propensity-scores}. By the above argument, we have:

  \[
    \begin{aligned}
      \abs*{\frac{\partial \overline{V}_i(\beta, \alpha)}{\partial \alpha}} & < C_{\pi}kM_i + \overline{C}_{\pi}C_{\pi}kM_i^2
    \end{aligned}
  \]
  Then, we have
  \[
\sup_{(\beta,\alpha)\in\mathcal{B}_{0}(\epsilon)} \abs*{\overline{V}_{i\alpha}(\beta, \alpha)} < \widetilde{M}_i = (C_{\pi}^2k^2 + \overline{C}_{\pi}C_{\pi}k)M^2_i + C_{\pi}kM_i
\]
and
\[
  \begin{aligned}
    \E[\widetilde{M}_i^{4}] &= C_1\E[M^8_i] + C_2\E[M_i^7] + C_3\E[M_i^6] + C_4\E[M_i^5] + C_5\E[M_i^4],
  \end{aligned}
\]
where $C_1$ through $C_5$ are constants that depend on $C_{\pi}$, $\overline{C}_{\pi}$, and $k$. By Assumption~\ref{a:bounded-moments}, each of the expectations on the right-hand side is uniformly bounded in $N$, which implies that $\E[\widetilde{M}_i^{4}]$ is also uniformly bounded in $N$.
\end{proof}